%% file: main.tex
\documentclass[11pt]{article}
\usepackage{xurl} 
\usepackage[utf8]{inputenc}
\usepackage{comment}
\usepackage{xcolor}
\usepackage{newtxtext}  
\usepackage[english]{babel}
\usepackage[T1]{fontenc} 
\usepackage[letterpaper,top=2cm,bottom=2cm,left=3cm,right=3cm,marginparwidth=1.75cm]{geometry}

\usepackage{amsmath}
\usepackage{amssymb}
\usepackage{amsthm}
\usepackage{graphicx}
\usepackage{physics}
\usepackage[colorlinks=true, allcolors=blue]{hyperref}
\usepackage{booktabs}
\usepackage{multirow}
\usepackage{tabularx}
\usepackage{natbib}
\usepackage{thmtools}
\usepackage{thm-restate}
\usepackage{authblk}
\usepackage{algorithm}
\usepackage{algpseudocode}
\algnewcommand{\Input}[1]{\item[\textbf{Input:}] #1}
\algnewcommand{\Output}[1]{\item[\textbf{Output:}] #1}

\algrenewcommand{\algorithmicrequire}{\textbf{Input:}}
\algrenewcommand{\algorithmicensure}{\textbf{Output:}}

\DeclareMathOperator{\dist}{dist}
\DeclareMathOperator{\corr}{corr}
\DeclareMathOperator{\E}{\mathbf{E}}

\usepackage{zongbo}
\newcommand{\dy}{\mathcal{D}_{\text{yes}}}
\newcommand{\dn}{\mathcal{D}_{\text{no}}}

\newcommand{\fy}{f_{\text{yes}}}
\newcommand{\fn}{f_{\text{no}}}

\newcommand{\Z}{\mathbb{Z}}

\def\inf{\mathbf{Inf}}
\def\supp{\textrm{supp}}
\def\Zfn{\ensuremath{\ZZ_4^n}}

\def\wU{\ensuremath{\widehat{U}}}
\def\diag{\textrm{diag}}

\title{\textbf{Efficient Non-Adaptive Quantum Algorithms for Tolerant Junta Testing}}

\author[1,\thanks{zongbo.bao@cwi.nl}]{Zongbo Bao}
\author[2,3,\thanks{liuyuxuan23s@ict.ac.cn}]{Yuxuan Liu}
\author[4,5,\thanks{phyao1985@gmail.com}]{Penghui Yao}
\author[6,\thanks{yezekun@fzu.edu.cn}]{Zekun Ye}
\author[2,3,\thanks{zhangjialin@ict.ac.cn}]{Jialin Zhang}

\affil[1]{CWI and Qusoft}
\affil[2]{State Key Lab of Processors, Institute of Computing 
  Technology, Chinese Academy of Sciences, Beijing 100190, China}
\affil[3]{School of Computer Science and Technology, University of Chinese 
  Academy of Sciences, Beijing 100049, China}  
\affil[4]{State Key Laboratory for Novel Software Technology, New Cornerstone Science Laboratory, Nanjing University, China}
\affil[5]{Hefei National Laboratory, Hefei 230088, China}
\affil[6]{College of Computer and Data Science, Fuzhou University, Fuzhou 350116, China}

\begin{document}
\maketitle

\begin{abstract}
    We consider the problem of deciding whether an $n$-qubit unitary (or $n$-bit Boolean function) is $\varepsilon_1$-close to some $k$-junta or $\varepsilon_2$-far from every $k$-junta, where $k$-junta unitaries act non-trivially on at most $k$ qubits and as the identity on the rest, and $k$-junta Boolean functions depend on at most $k$ variables. For constant numbers $\varepsilon_1,\varepsilon_2$ such that $0 < \varepsilon_1 < \varepsilon_2 < 1$, we show the following.

    \begin{enumerate}
        \item A non-adaptive $O(k\log k)$-query tolerant $(\varepsilon_1,\varepsilon_2)$-tester for $k$-junta unitaries when  $2\sqrt{2}\varepsilon_1 < \varepsilon_2$. 

        \item A non-adaptive tolerant $(\varepsilon_1,\varepsilon_2)$-tester for $k$-junta Boolean functions with $O(k \log k)$ quantum queries when  $4\varepsilon_1 < \varepsilon_2$.
        
        \item A $2^{\widetilde{O}(k)}$-query tolerant $(\varepsilon_1,\varepsilon_2)$-tester for $k$-junta unitaries for any $\varepsilon_1,\varepsilon_2$.
    \end{enumerate}
    The first algorithm provides an exponential improvement over the best-known quantum algorithms \cite{chen2024tolerant,arunachalam2025testing}. The second algorithm shows an exponential quantum advantage over any non-adaptive classical algorithm \cite{chen2024mildly}. The third tester gives the first tolerant junta unitary testing result for an arbitrary gap. 

    Besides, we adapt the first two quantum algorithms to be implemented using only single-qubit operations, thereby enhancing experimental feasibility, with a slightly more stringent requirement for the parameter gap.
\end{abstract}

\section{Introduction}
\input{sections/1_introduction}

\section{Preliminary}
\input{sections/2_preliminary}

\section{High-Influence Coordinate Extractor}\label{sec:ce}
\input{sections/3_reduction}

\section{Efficient Tolerant Junta Tester with Constant Factor Gap}\label{sec:gap}
\input{sections/4_gapped_tester}

\section{Tolerant Junta Tester with Constant Gap using Single-Qubit Operations
}\label{sec:local}
\input{sections/5_gapped_tester_local}

\section{Tolerant Junta Tester with Arbitrary Gap}\label{sec:gapless}
\input{sections/6_ungapped_tester}

\section{Conclusion}
\input{sections/7_conclusion}

\appendix
\section{Lower Bound for Classical 
  Tolerant \texorpdfstring{$k$}{k}-Junta Tester}
  \label{sec:app:lowerbound}
\input{sections/app1_lowerbound}

\section{Proof of \texorpdfstring{\Cref{lemma:is}}{is}}
\label{appendixB}
\input{sections/app2_proof_of_influence_sample}

\subsection*{Acknowledgments}

This work was done when Z.Y. was a Ph.D. student at Nanjing University. This work was supported by the Innovation Program for Quantum Science and Technology (Grant No. 2024ZD0300500), the National Natural Science Foundation of China (Grants No. 62272441, 62332009, 12347104), the Dutch Ministry of Economic Affairs and Climate Policy (EZK), as part of the Quantum Delta NL programme (KAT2), NSFC/RGC Joint Research Scheme (Grant no.12461160276) and Natural Science Foundation of Jiangsu Province (No. BK20243060).

\bibliography{ref}
\end{document}

%% file: sections/1_introduction.tex
\emph{Property testing} \cite{goldreich1998property} constitutes a fundamental research field in theoretical computer science, aiming to determine whether an object satisfies a certain global property or is far from it using only a sublinear number of local queries or samples. Its advantage emerges when handling massive inputs: instead of reading and processing the entire dataset to obtain an exact description, it employs randomized queries whose number is substantially smaller than the input size to extract key information and achieve a probabilistically sound conclusion of the target property. Since the 1990s, property testing has been extensively investigated for various objects, including functions \cite{blum1990self, rubinfeld1996robust}, graphs \cite{goldreich1997property, alon2000efficient}, probability distributions \cite{batu2013testing}, etc.

The advent of quantum computation has introduced a powerful new paradigm for property testing \cite{buhrman2008quantum, montanaro2016survey}. By leveraging quantum superposition and entanglement, the quantum query model enables access to multiple inputs simultaneously, achieving significant reductions in query complexity. Foundational results like Grover's algorithm \cite{grover1996fast} and Shor's algorithm \cite{shor1994algorithms} demonstrate the advantage of the quantum query model over its classical counterpart. Consequently, applying this quantum advantage to the task of property testing has given rise to \emph{quantum property testing}, an independent research direction that can not only test classical properties more efficiently \cite{friedl2003quantum} but also verify quantum objects \cite{o2015quantum} such as quantum states and operations. Notably, quantum advantage exhibits substantial variation across different properties: some problems have polynomial speedups while others achieve exponential improvements. This disparity has prompted an in-depth exploration into the nature of quantum query complexity.

The property we are interested in here is being a \emph{k-junta}, first defined for Boolean functions whose output depends on only $k$ out of its $n$ input variables. Efficient junta testing algorithms have broad implications in computational complexity theory and related areas such as learning theory and data analysis, critical for identifying relevant features in high-dimensional datasets. Later, Montanaro and Osborne \cite{montanaro2008quantum} extended the definition to quantum systems, noting that an $n$-qubit unitary $U$ is a $k$-junta if it acts non-trivially on only $k$ of the $n$ qubits (see \Cref{def:junta} for a formal definition). Similarly, testing $k$-junta unitary reveals the underlying structure and behavior of quantum systems by identifying the crucial subset of qubits, thus enabling optimized deployment of quantum gates and targeted error correction. Such focused operations significantly improve efficiency and robustness in the NISQ era. 

The ``standard'' testing framework discussed previously requires algorithms to strictly distinguish whether the target object (a Boolean function or a unitary operator) is exactly a $k$-junta or far from any $k$-junta. However, if noise or errors create imperfect $k$-juntas, we still need the robust identification of functions that approximately depend on a small set of critical variables. To address this challenge, Parnas, Ron, and Rubinfeld \cite{parnas2006tolerant} pioneered the notion of \emph{tolerant property testing}. For $k$-junta testing, the algorithm must now distinguish whether the object is sufficiently close to some $k$-junta or far from all $k$-juntas. This notion of tolerant testing is arguably even more critical in the quantum setting than in its classical counterpart. While a classical object can, in principle, perfectly represent a $k$-junta, the unavoidable noise inherent in any realistic quantum device in the current NISQ era means that a tested quantum object is fundamentally imperfect. Consequently, for practical applications in the era of fault-tolerant quantum computing, the tolerant model is not merely a theoretical alternative but a necessity. The tolerant testing model carries greater practical significance than its standard counterpart but also presents substantial challenges. Existing research demonstrates a remarkable increase in difficulty \cite{chen2023new, chen2024mildly}, with current approaches \cite{chen2024tolerant,chen2025mysterious} yet to break through the theoretical barrier of exponential query complexity.

While classical tolerant testing has been studied extensively, its quantum counterparts are less explored, despite the potential for quantum speedups in efficiency.
In this work, we aim to construct efficient quantum algorithms for the problem of tolerant junta testing for both Boolean functions and unitaries. We formally address the following problem:

\begin{problem}[Tolerant Junta Testing]\label{prob} 
Given oracle access to an $n$-qubit unitary $U$ (or a Boolean function $f:\{\pm 1\}^n \rightarrow \{\pm1\}$), for integer $ 1 \le k < n$ and $0<\vep_1<\vep_2< 1$, decide with high probability whether $U$(or $f$) is $\varepsilon_1$-close to some $k$-junta, or is $\varepsilon_2$-far from every $k$-junta.
\end{problem}

Here we employ the normalized Frobenius distance as defined in \Cref{def_dist_unitary}.

An algorithm for such a problem is called a \emph{tolerant $(\vep_1, \vep_2)$-tester} for $k$-juntas. To avoid misunderstanding, note that by tolerant \emph{quantum} junta testing, we mean using the quantum algorithm, and junta includes Boolean functions and unitaries. 

We resolve this problem by constructing three different quantum tolerant testers. Our main algorithms achieve a query complexity of $O(k \log k)$ for both Boolean functions and unitaries, provided there is a constant factor gap between $\varepsilon_1$ and $\varepsilon_2$. Then we replace the \algname{Fourier-Sample} with \algname{Influence-Sample} to implement the same algorithm using only single-qubit operations, thus making it easier to realize experimentally, with a slightly more demanding requirement for the parameter gap.
The last result is an algorithm for unitaries with $2^{\widetilde{O}\p{k}}$ queries that handles arbitrary distance parameters $\varepsilon_1 < \varepsilon_2$ with a different technical approach. Here $\widetilde{O}(\cdot)$ hide polylogarithmic factors.

\subsection{Related Work}

The problem of junta testing has a rich history across classical and quantum computation, in both the standard and tolerant settings. We summarize the key milestones in \Cref{tab:s,tab:t}. In \Cref{tab:t}, we assume that $\vep_1,\vep_2$ are constant.
        
\paragraph{Classical Standard Testing of $k$-Junta Boolean Functions}

The study of junta testing began with the $k=1$ case (also called dictator functions) in 2002, for which Parnas, Ron, and Samorodnitsky \cite{parnas2002testing} gave an $O(1)$-query tester. In 2004, Fischer, Kindler, Ron, Safra, and Samorodnitsky \cite{fischer2004testing} proposed the first general $k$-junta tester with query complexity $\widetilde{O}(k^2)$, which was independent of the ambient dimension $n$. This is a non-adaptive algorithm, meaning its query strategy is fixed in advance and independent of previous query results. In contrast, adaptive algorithms dynamically adjust their subsequent queries based on prior query results. 

In 2008, Blais \cite{blais2008improved} improved the non-adaptive query complexity to $\widetilde{O}(k^{3/2})$, and then developed an adaptive algorithm with $O(k \log k)$ queries \cite{blais2009testing}. In 2018, separate proofs for non-adaptive \cite{chen2018settling} and adaptive \cite{sauglam2018near} lower bounds demonstrated that these results are indeed optimal.

\paragraph{Quantum Standard Testing of $k$-Junta Boolean Functions}

In 2007, Atıcı and Servedio \cite{atici2007quantum} designed the first quantum algorithm using $O(k)$ quantum examples, whose core technique was \emph{Fourier sampling} \cite{bshouty1995learning}. 
This technique became a cornerstone for later quantum junta testers. In 2016, Ambainis, Belovs, Regev, and de Wolf \cite{ambainis2016efficient} proposed a quantum algorithm for testing juntas with only  $\widetilde{O}(\sqrt{k})$ queries.
A matching $\widetilde{\Omega}(\sqrt{k})$ lower bound was subsequently proved by Bun, Kothari, and  Thaler \cite{bun2018polynomial}, establishing a quadratic speedup over the optimal classical algorithm.

\paragraph{Classical Tolerant Testing of $k$-Junta Boolean Functions}

\textbf{For non-adaptive algorithms}, in 2016 Blais, Canonne, Eden, Levi, and Ron~\cite{blais2018tolerant} proposed a $(\frac{\rho}{16}\vep, \vep)$-tester using $O\big(\frac{k\log k}{\vep \rho(1-\rho)^k}\big)$ queries, and requires $\vep_2 \geq 16\vep_1$, demonstrating a trade-off between tolerant and query complexity. In 2019, De, Mossel, and Neeman \cite{de2019junta} designed another algorithm with the query complexity of $2^{k}\poly(k)$, which not only worked for arbitrary parameters but also achieved a stronger goal of identifying the closest $k$-junta to the target function. In 2024, Nadimpalli and Patel \cite{nadimpalli2024optimal} reformulated the distance to a $k$-junta as an expectation and estimated it by sampling a local Hamming ball, reducing the query complexity to $2^{\widetilde{O}(\sqrt{k})}$.

Regarding the lower bound, Levi and Waingarten \cite{levi2019lower} proved that any non-adaptive tolerant testing of $k$-juntas requires at least $\widetilde{\Omega}(k^2)$ queries. Afterward, Chen, De, Li, Nadimpalli, and  Servedio \cite{chen2024mildly} employed an elegant construction of two distributions on Boolean functions and applied Yao's minimax principle to establish an improved lower bound of $2^{\Omega(\sqrt{k})}$, matching the upper bound above. 

\textbf{For adaptive algorithms}, Iyer, Tal, and Whitmeyer \cite{iyer2021junta} proposed an algorithm with query complexity $2^{\widetilde{O}(\sqrt{k})}$. Recently, Chen, Patel, and Servedio \cite{chen2025mysterious} made a breakthrough by connecting tolerant junta testing with the agnostic learning problem in computational learning theory, obtaining an adaptive $k$-junta tolerant tester with query complexity $2^{\widetilde{O}(k^{1/3})}$. Since the lower bound for non-adaptive $k$-junta testing is $2^{\Omega(\sqrt{k})}$, their result shows that adaptive algorithms outperform non-adaptive ones in this problem.

Additionally, concerning adaptive lower bounds, Chen and Patel \cite{chen2023new} proved a $k^{-\Omega(\log(\vep_2-\vep_1))}$ lower bound. Bridging the vast gap and pinpointing the true adaptive complexity of tolerant junta testing remains a central open problem in the field.

\paragraph{Quantum Standard Testing of $k$-Junta Unitaries}

In 2008, Montanaro and Osborne \cite{montanaro2008quantum} initiated the study of Fourier analysis over the space of matrices, which is later referred to as Pauli analysis~\cite{10.1145/3618260.3649662}. They also gave a tester for 1-junta.
Wang \cite{wang2011property} in 2011 extended the techniques for Boolean functions in \cite{atici2007quantum} to test $k$-junta unitaries with $O(k)$ queries. Chen, Nadimpalli, and Yuen \cite{chen2023testing} advanced the methodology of \cite{ambainis2016efficient} and established tight upper and lower bounds, demonstrating that the quantum query complexity for this problem is also $\widetilde{\Theta}(\sqrt k)$.

A more general quantum subject is quantum channels, which represent general linear transformations in quantum systems that are not necessarily reversible. Bao and Yao \cite{bao2023testing} investigated the $k$-junta testing problem for quantum channels and constructed an algorithm using $O(k)$ queries, while proving a lower bound of $\Omega(\sqrt k)$.

\paragraph{Quantum Tolerant Testing of $k$-Junta Unitaries}

The problem of tolerant testing for unitaries is largely unexplored. Recently, Chen, Li, and Luo \cite{chen2024tolerant} proposed a non-adaptive $(\frac{\sqrt \rho}{8}\vep, \vep)$-tester with $O\big(\frac{k\log k}{\vep^2 \rho(1-\rho)^k}\big)$ queries, which extended upon the classical algorithm for Boolean functions \cite{blais2018tolerant}. From a simplicity standpoint, while this "quantization" of a classical method is a natural starting point, it carries over the inherent complexities of the original approach. Moreover, their tester requires a large gap between the distance parameters, i.e., $\vep_2 \geq 8\vep_1$, and when this factor is a constant, their algorithm needs exponential queries. 

Meanwhile, Arunachalam, Dutt, and Gutiérrez \cite{arunachalam2025testing} introduced a tolerant $(\vep_1, \vep_2)$-tester with $O(16^k/(\vep_2-\vep_1)^4)$ queries with an $n$-qubit quantum memory, and $O(16^k \cdot n/(\vep_2-\vep_1)^4)$ when implemented with a novel memory-less subroutine if $\vep_2 \ge 2 \vep_1$. Their exponential query complexity highlights the need for a new approach, one that can circumvent this exponential barrier and offer greater generality, motivating the results we present in this paper.

\begin{table}[ht]
    \setcellgapes{5pt}
    \makegapedcells
    \renewcommand\arraystretch{1.5}
    \centering
    \caption{Standard testing of $k$-juntas}
    \label{tab:s}
    \scalebox{0.86}{
    \begin{tabular}[t]{lccc}
    \hline
    & Classical & Classical & Quantum \\
    & non-adaptive & adaptive &  \\
    \hline
    Boolean function& \makecell{ $\widetilde{O}\big(k^{3/2})$ \vspace{2pt}\\~(Blais, 2008) \cite{blais2008improved}} &
    \makecell{ $O\big(k\log k\big)$ \vspace{2pt}\\~(Blais, 2009) \cite{blais2009testing} } & 
    \makecell{$\widetilde{O}\big(\sqrt{k}\big)$ \vspace{1pt}\\~(Ambainis et al. 
    2016) \cite{ambainis2016efficient} }  \\
    &
    \makecell{$\widetilde{\Omega}\big(k^{3/2})$ \vspace{2pt}\\~(Chen et al. 
    2018) \cite{chen2018settling}} &
    \makecell{ $\Omega\big(k\log k\big)$ \vspace{2pt}\\~(Sa{\u g}lam, 2018) \cite{sauglam2018near} } & 
    \makecell{$\widetilde{\Omega}\big(\sqrt{k}\big)$ \vspace{1pt}\\~(Bun et al. 2018) \cite{bun2018polynomial} }  \\
     unitary & --- & 
    --- & 
    \makecell{$\widetilde{O}\big(\sqrt{k}\big)$ \vspace{1pt}\\~(Chen et al. 2023) \cite{chen2023testing} }  \\
    & --- & 
    --- & 
    \makecell{$\Omega\big(\sqrt{k}\big)$ \vspace{1pt}\\~(Chen et al. 2023) \cite{chen2023testing} } \\    
    \hline
    \end{tabular}}
\end{table}

\begin{table}[ht]
    \setcellgapes{5pt}
    \makegapedcells
    \renewcommand\arraystretch{1.5}
    \centering
    \caption{Tolerant testing of $k$-juntas}
    \label{tab:t}
    \scalebox{0.86}{
    \begin{tabular}[t]{lccc}
    \hline
    & Classical & Classical & Quantum \\
    & non-adaptive & adaptive &  \\
    \hline
    Boolean function& \makecell{ $2^{\widetilde{O}(\sqrt k)}$ \vspace{2pt}\\~(Nadimpalli et al. 2024) \cite{nadimpalli2024optimal}} &
    \makecell{ $2^{\widetilde{O}(k^{1/3})}$ \vspace{2pt}\\~(Chen et al. 2025) \cite{chen2025mysterious} } & 
    \makecell{$O(k \log k), \vep_2 > 4\vep_1$ \vspace{2pt}\\ \textbf{~\Cref{thm:tester-boolean} in this work } } \\
    &
    \makecell{$2^{\Omega(\sqrt k)}$ \vspace{2pt}\\~(Chen et al. 2024) \cite{chen2024mildly}} &
    \makecell{ $k^{\Omega(1)}$ \vspace{2pt}\\~(Chen et al. 2023) \cite{chen2023new} } & 
    \makecell{unexplored }  \\
     unitary & --- & 
    --- & 
    \makecell{$O(k \log k), \vep_2 > 2\sqrt 2\vep_1$ \vspace{2pt}\\\textbf{~\Cref{thm:tester} in this work} \\ $2^{\widetilde{O}(k)}$ \vspace{1pt}\\\textbf{~\Cref{thm:general} in this work} }   \\
    & --- & 
    --- & 
    \makecell{unexplored} \\    
    \hline
    \end{tabular}}
\end{table}

\subsection{Our Results}

It is noteworthy that while tolerant testing of $k$-junta Boolean functions has been extensively studied in the classical scenario, no quantum algorithm has been proposed to date. A possible explanation is that the query complexity of classical algorithms is exponential, and researchers do not expect quantum algorithms to achieve an exponential speedup for this problem. If only polynomial speedup is possible, the quantum query complexity would remain exponential, offering no fundamental improvement in solving efficiency.

However, our recent investigation represents a significant breakthrough, suggesting that there exists a concise tester with only $O(k\log k)$ quantum queries. This approach can apply to both Boolean functions and unitary operators, though it comes with the limitation of requiring a certain multiplicative gap between $\vep_1$ and $\vep_2$ ($4$ for Boolean functions and $2\sqrt 2$ for unitaries). 

Specifically, for unitaries, the $O(k\log k)$ complexity in \Cref{thm:i2} offers an exponential improvement over the previous best results \cite{chen2024tolerant, arunachalam2025testing}, which requires $\exp(k)$ queries when there is a constant factor gap. For Boolean functions, our algorithm provides the first non-trivial quantum tester for this problem, as shown in \Cref{thm:i1}. In the following, we assume $\vep_1,\vep_2$ are constant. 

\begin{theorem}
[Tolerant Tester for Unitaries]
\label{thm:i2}
   There exists a non-adaptive tolerant $(\vep_1, \vep_2)$-tester solving \Cref{prob} for $k$-junta unitaries with $O\p{k \log k}$ quantum queries if $\vep_2 > 2\sqrt{2}\vep_1$.
\end{theorem}

\begin{theorem}
[Tolerant Tester for Boolean Functions]
\label{thm:i1}
    There exists a non-adaptive tolerant $(\vep_1, \vep_2)$-tester solving \Cref{prob} for $k$-junta Boolean functions with $O\p{k \log k}$ quantum queries if $\vep_2 > 4\vep_1$.
\end{theorem}

In \Cref{sec:local}, we adapt our algorithm by replacing the \algname{Fourier-Sample} with \algname{Influence-Sample}. This modification allows the algorithm to be implemented using only single-qubit operations while eliminating the need for CNOTs, thereby enhancing its experimental feasibility. The trade-off, however, is a slightly more stringent requirement for the parameter gap.

The \algname{Influence-Sample} also eliminates the need for quantum memory (i.e., ancilla qubits) in our algorithm. For comparison, the tester in \cite{arunachalam2025testing} needs $O(4^kn)$ queries without quantum memory, whereas our algorithm requires only $O\p{k \log k}$ queries in this setting.

\begin{theorem}
[Tolerant Tester for Unitaries with Single-Qubit Operations]
\label{thm:i4}
   There exists a non-adaptive tolerant $(\vep_1, \vep_2)$-tester that only requires single-qubit operations solving \Cref{prob} for $k$-junta unitaries with $O\p{k \log k}$ quantum queries if $\vep_2 > 3\sqrt{2} \vep_1$. Moreover, the tester does not require ancilla qubits.
\end{theorem}

\begin{theorem}
[Tolerant Tester for Boolean Functions with Single-Qubit Operations]
\label{thm:i3}
   There exists a non-adaptive tolerant $(\vep_1, \vep_2)$-tester that only requires single-qubit operations solving \Cref{prob} for $k$-junta Boolean functions with $O\p{k \log k}$ quantum queries if $\vep_2 > 6\vep_1$. Moreover, the tester does not require ancilla qubits.
\end{theorem}

Here, we obtain an exponential separation 
  in query complexity
  between quantum and classical tolerant 
  junta testing, combining the fact as follows.
For self-containment, we provide the proof of
  the lower bound in \Cref{sec:app:lowerbound},
  we note that the proof also follows the similar line as 
  \cite{pallavoor2022approximating}.

\begin{restatable}[Classical lower bound, Theorem 2 of \cite{chen2024mildly}, 
  modified]{fact}{lowerboundI}
\label{thm:lowerbound1}
  Any non-adaptive classical $(0.01, 0.2)$-tolerant tester for $k$-junta Boolean functions  with
    success probability at least $9/10$
    must make $2^{\Omega\p{k^{1/2}}}$ queries.
\end{restatable}

According to \Cref{thm:i3}, our quantum 
  tolerant $k$-junta testers for Boolean functions can
  achieve a success probability of $9/10$ with 
  $O\p{k\log k}$ queries. 
An exponential separation on the query complexity follows.
Also, note that in this separation, the quantum algorithm
  only uses single-qubit operations.
According to \Cref{lemma:is}, we can further obtain
  a quantum algorithm
  achieving exponential saving using only Hardmard 
  gates, computational basis measurements, and classical
  post-processing.

Another of our results is a $2^{\widetilde{O}(k)}$-query algorithm for tolerant testing of unitaries, with query complexity comparable to that in \cite{chen2024tolerant,arunachalam2025testing}. However, our algorithm is more general as it applies to any parameters $0 < \vep_1 < \vep_2 < 1$, and employs a fundamentally distinct approach.

\begin{theorem}[General Tolerant Tester for Unitaries]\label{thm:i5}
   For any $k,\vep_1,\vep_2$, there exists a non-adaptive tolerant $(\vep_1, \vep_2)$-tester solving \Cref{prob} for $k$-junta unitaries with $2^{O\p{k\log k/\vep}}$ quantum queries, where $\vep \coloneqq \vep_2-\vep_1$.
\end{theorem}

It is noted that the testers in \Cref{thm:i2,thm:i1,thm:i4,thm:i3} only require access to the unitary $U$ itself, while the tester in \Cref{thm:i5} requires controlled applications of $U$.

\subsection{Technical Overview}

Our algorithms employ the theory of Pauli analysis, which is Fourier analysis on the space of operators and super-operators~\cite{montanaro2008quantum,10.1145/3618260.3649662}. Pauli analysis has received increasing attention in the past couple of years, and has found applications in quantum learning theory, quantum circuit lower bounds, and quantum computational complexity~\cite{10.1145/3618260.3649662,10.1145/3717823.3718189,arunachalam_et_al:LIPIcs.ICALP.2024.13}.

Here we give a brief overview of our algorithms.

\paragraph{Fourier-Sample}
As a quantum query algorithm, the only quantum subroutine in our technique is \algname{Fourier-Sample} (\Cref{alg:fs}) introduced in \cite{bshouty1995learning}. It utilizes the Choi-Jamiolkowski isomorphism to encode a unitary $U$ as a quantum state and measure it in the Pauli basis, enabling us to sample a Pauli string with the probability of its Fourier weight. This is a simple yet fundamental subroutine that is used repeatedly throughout our algorithm. Simulating this process classically requires an exponential number of queries.

\paragraph{Coordinate-Extractor}
At the beginning, we propose an essential approach to reduce the scale of our problem. It is based on the following connection: a unitary that is close to $k$-junta should have most of its influence concentrated on just $k$ qubits. So in \Cref{sec:ce} we design a subroutine \algname{Coordinate-Extractor} (\Cref{alg:ce}) to efficiently identify a small number of coordinates with high influence in the function or unitary. 

The basic strategy is using \algname{Fourier-Sample} to estimate the influence of all $n$ coordinates simultaneously.
During this procedure, we need to call \algname{Fourier-Sample} $O(k \log k)$ times. For each sample with low weight, we use it as evidence for the influence of every coordinate in its support. The key to efficiency is that each quantum sample is reused to update the estimates for all relevant coordinates. 
Next, we construct a small set $S$ with $|S|=O(k^2)$ that contains coordinates whose \emph{degree-$k$ influence} (see \Cref{eq:k-inf} and \Cref{eq:k-inf-u} for the definition) exceeds a certain threshold. 
Importantly, we prove that $S$ is sufficient for approximating the distance to the closest $k$-junta. That is to say, the distance from $U$ to the class of all $k$-juntas can be well-approximated by its distance to the class of $k$-juntas that only act on the small set $S$. Therefore, the number of coordinates we need to consider was reduced from $n$ to $O(k^2)$, making the subsequent tolerant testing step independent of the input size $n$.

\paragraph{Our First Tester with Constant Gap}
In \Cref{sec:gap}, we present our first main algorithm, a highly efficient tolerant quantum junta tester for both Boolean functions and unitaries, achieving an $O(k\log k)$ query complexity. The strategy is founded on the tight relationship between the distance to the class of $k$-juntas and its influence: functions close to a $k$-junta on a set $T$ have very little influence outside $T$, while functions far from all $k$-juntas have significant influence outside any choice of $T$.

The core testing process is straightforward (see \Cref{alg:tester}). As a preprocessing step, we first employ \algname{Coordinate-Extractor} to identify a small candidate set $S$. Then for every possible subset $T \subseteq S$ with size $k$, we repeatedly invoke \algname{Fourier-Sample} and count the fraction of outcomes whose support is not contained within $T$. It provides the estimation for $\inf_T[U]$ for every $T$, enabling us to reliably distinguish between YES and NO instances.

\paragraph{An Experiment-friendly Tester}
Although \algname{Fourier-Sample} is straightforward and intuitive, its implementation in experiments requires multi-qubit entanglement using CNOT gates, which is difficult to execute and introduces significant errors. Therefore, in \Cref{sec:local} we replace \algname{Fourier-Sample} with \algname{Influence-sample} from \cite{bao2023testing} in our algorithm (\Cref{alg:tester-l}). This alternative approach can also estimate the influence using only single-qubit operations on single qubits. This is more experimentally feasible, but comes at the cost of a stricter requirement for the parameter gap.

\paragraph{Another Tester for Arbitrary Gap}
The technique of our final tester for unitaries in \Cref{sec:gapless}, which can be applied for arbitrary parameters, is inspired by existing work on Boolean functions \cite{nadimpalli2024optimal}. In analogy with their approach, which represents the distance from a function to $k$-junta as the expectation of a function, we formulate the distance from a unitary operator to $k$-junta as the nuclear norm of a matrix via the partial trace (see \Cref{eq:dist}).

The subsequent task is to estimate this norm. It can be shown that the nuclear norm of a matrix is a continuous function of its entries. Therefore, after again applying the \algname{Coordinate-Extractor} to reduce the query complexity, we directly estimate each element of the matrix using the Hadamard test \cite{luongo2022quantum}. See \Cref{alg:tester2} for the details of our approach.

%% file: sections/2_preliminary.tex
In this section, we present some basic definitions and notations. For a positive integer $n$, we use $[n]$ to denote the set $\{1, 2, \dots , n\}$. For any subset $T \subseteq [n]$, we denote its complement as $\overline{T}  \coloneqq [n]\backslash T$. We use $|T|$ to denote the cardinality of $T$. For two strings $x \coloneqq(x_1, x_2, \dots, x_n), y \coloneqq(y_1, y_2, \dots, y_n)$ and $T \subseteq [n]$, we use $x^{\overline{T}}y^T$ to denote the string $z \coloneqq(z_1, z_2, \dots, z_n)$ in which $z_i = x_i$ for $i\in {\overline{T}}$ and $z_i = y_i$ for $i \in T$. For a matrix $A_{N\times N}$ and $0 \leq i,j <N$, denote the element in row $i$ and column $j$ of $A$ as $A[i][j]$.

\subsection{Boolean Functions}

Fourier analysis of Boolean functions was pioneered by Kahn, Kalai, and Linial, who established the celebrated KKL theorem  \cite{kahn1989influence} and paved the way for many fundamental results in Boolean functions. Later, these techniques were systematically consolidated by O'Donnell \cite{o2014analysis}. Most concepts are initially defined for any real-valued Boolean functions, but here we only clarify them for Boolean-valued ones.

Let $f \colon \{\pm 1\}^n \to \{\pm 1\}$ be a Boolean function. For any subset $S \subseteq [n]$, $\chi_S(x)  \coloneqq \prod_{i \in S} x_i$ are the \emph{parity basis functions}, and the \emph{Fourier expansion} of $f$ is its unique representation as a multilinear polynomial:
\[ f(x) = \sum_{S \subseteq [n]} \widehat{f}(S) \chi_S(x), \]
where $\widehat{f}(S)$ are called the \emph{Fourier coefficients} of $f$ corresponding to subset $S$, computed by
\[
\widehat{f}(S) = \mathbb{E}_x[f(x)\chi_S(x)].
\]
We have Plancherel's and Parseval's formulas:
\[
\langle f, g \rangle= \mathbb{E}_x[f(x) g(x)] = \sum_{S \subseteq [n]} \widehat{f}(S) \widehat{g}(S)
\quad \text{and} \quad
\|f\|^2 = \sum_{S \subseteq [n]} \widehat{f}^2(S)=1.
\]
therefore $\widehat{f}^2(S)$ is usually called the Fourier weight of $S$ on $f$.

We next define the influence of variables on a Boolean function. It quantifies the importance of a variable's value to the function's output. Two common definitions are given below, where the first involving flipping a certain bit is more intuitive, while the second related to Fourier weights plays an important role in our algorithms.

\begin{definition}[Influence of variables]\label{def_boolean_inf}
For $f \colon \{\pm 1\}^n \to \{\pm 1\}$ and any $i \in [n]$, the \emph{influence} of variable $i$ on $f$ is defined as
\[
\inf_i(f) \coloneqq \Pr_{x \in \{\pm 1\}^n}\left[f(x) \neq f(x^{\oplus i})\right]= \sum_{\substack{S \subseteq [n], S \ni i}} \widehat{f}^2(S),
\]
where $x^{\oplus i}$ differs from $x$ exactly in the $i$-th position. 
For $k \leq n$, we define the \emph{degree-$k$ influence} of variable $i$ on $f$ as
\begin{equation}\label{eq:k-inf}
\inf^{\leq k}_i(f) \coloneqq \sum_{\substack{S \subseteq [n], S \ni i, |S| \leq k}} \widehat{f}^2(S).
\end{equation}
The \emph{influence} of the set $T \subseteq [n]$ on $f$ is defined as
\begin{equation*}
    \inf_T(f) \coloneqq 2 \Pr_{x,y}\left[f(x) \neq f(x^{\overline{T}}y^T)\right] = \sum_{\substack{S \subseteq [n], S \cap T \neq \emptyset}} \widehat{f}^2(S).
\end{equation*}

\end{definition}

\begin{definition}[Distance between Boolean functions]\label{def_dist_boolean}
The distance between Boolean functions $f,g:\{\pm1\}^n \rightarrow \{\pm1\}$ is defined as
\[
\dist(f,g)  \coloneqq \Pr_{\substack{x \sim \{\pm 1\}^n}} [f(x) \neq g(x)].
\]
For any $T \subseteq [n]$, let $\mathcal{J}_T$ denote the class of juntas on $T$. For $S \subseteq [n]$ with $|S| \ge k$, let $\mathcal{J}_{S,k}$ denote the class of $k$-juntas on $S$. Then we define
        \[
\dist(f,\mathcal{J}_{T}) \coloneqq\min_{g \in \mathcal{J}_T} \dist(f, g),
        \]
        and
        \[
\dist(f,\mathcal{J}_{S,k}) \coloneqq\min_{T \subseteq \binom{S}{k}} \dist(f, \mathcal{J}_T).
        \]
         For simplicity, we abbreviate $\mathcal{J}_{[n],k}$ as $\mathcal{J}_{k}$. 
\end{definition}

\begin{remark}
    
For a class $\mathcal{C}$ of Boolean functions, if we can estimate $\dist(f, \mathcal{C})\coloneqq\min_{g \in \mathcal{C}} \dist(f,g)$ up to additive error $(\varepsilon_2 - \varepsilon_1)/2$ with high probability, then the estimation algorithm is a tolerant $(\vep_1, \vep_2)$-tester for that class.

\end{remark}

\begin{definition}
    [Correlation between Boolean functions]\label{def_corr_boolean}
    The correlation between Boolean functions $f,g:\{\pm 1\}^n \rightarrow \{\pm1\}$ is defined as
\[
\corr(f,g)  \coloneqq \langle f,g \rangle =  \Pr_{\substack{x \sim \{\pm 1\}^n}} [f(x)g(x)].
\]
For any $T\subseteq [n]$, let 
        \[
\corr(f,\mathcal{J}_{T}) \coloneqq\max_{g \in \mathcal{J}_T} \corr(f, g).
        \]
        For $S \subseteq [n]$ with $|S|\ge k$, let
        \[
\corr(f,\mathcal{J}_{S,k}) \coloneqq\max_{T \subseteq \binom{S}{k}} \corr(f, \mathcal{J}_T),
        \]
          where $\mathcal{J}_T$ and $\mathcal{J}_{S,k}$ are defined as \Cref{def_dist_boolean}.
\end{definition}
By \Cref{def_dist_boolean,def_corr_boolean}, for any $T \subseteq [n]$, we have
\begin{equation}\label{eq_relation_corr_dist}
\corr(f,\mathcal{J}_T) = 1-2\dist(f,\mathcal{J}_T).
\end{equation}

\subsection{Unitary Operators}

Recall that the Pauli operators are defined as
\[
\sigma_0 = \begin{pmatrix} 1 & 0 \\ 0 & 1 \end{pmatrix} = I,\quad
\sigma_1 = \begin{pmatrix} 0 & 1 \\ 1 & 0 \end{pmatrix} = X,\quad
\sigma_2 = \begin{pmatrix} 0 & -i \\ i & 0 \end{pmatrix} = Y,\quad
\sigma_3 = \begin{pmatrix} 1 & 0 \\ 0 & -1 \end{pmatrix} = Z.
\]
$\left\{ \frac{1}{\sqrt{2}} \sigma_0 , \frac{1}{\sqrt{2}} \sigma_1 , \frac{1}{\sqrt{2}} \sigma_2 , \frac{1}{\sqrt{2}} \sigma_3 \right\}$ form an orthonormal basis of the four-dimensional complex vector space to the Hilbert-Schmidt inner product
\[
\langle A, B\rangle
\coloneqq\Tr(A^\dagger B).
\]

For $x \in \{0, 1, 2, 3\}^n \cong \Zfn$, define $\sigma_x  \coloneqq \sigma_{x_1} \otimes \cdots \otimes \sigma_{x_n}$ and $\supp(x)  \coloneqq \{i \in [n] : x_i \neq 0\}$. Let $N\coloneqq 2^n$. It is easy to check that $\left\{ \frac{1}{\sqrt{N}} \sigma_x \right\}_{x \in \Zfn}$ forms an orthonormal basis (called the Pauli basis) of the complex vector space of dimension $4^n$ to the Hilbert-Schmidt inner product. 

The following concepts work for any linear operators, but we only clarify them for unitary operators: under the \emph{Pauli basis}, any unitary operator can be uniquely represented as a multilinear polynomial:
\[
U = \sum_{x \in \Z_4^n} \widehat{U}(x) \sigma_x,
\]
where $\widehat{U}(x) = \frac{1}{N} \langle U, \sigma_x \rangle 
$. It is easy to verify Plancherel's and Parseval's formulas:
\[
\frac{1}{N}
\langle U, V \rangle
= \sum_{x \in \Zfn} \widehat{U}(x)^*  \widehat{V}(x)
\quad \text{and} \quad
\frac{1}{N} \|U\|^2 = \sum_{x \in \Zfn} \left|\widehat{U}(x)\right|^2=1.
\]

\begin{definition}[Influence of qubits]\label{def_unitary_inf}
For an $n$-qubit unitary operator $U$
and any $i \in [n]$, the \emph{influence} of variable $i$ on $U$ is defined as \[
\inf_i[U] \coloneqq \sum_{x \in \Zfn: x_i \neq 0} \abs{\wU(x)}^2.
\]
The \emph{degree-$k$ influence} of variable $i$ on $U$ is defined as
\begin{equation}\label{eq:k-inf-u}
\inf_i^{\le k}[U] \coloneqq \sum_{x \in \Zfn: |x| \le k,x_i \neq 0} \abs{\wU(x)}^2.
\end{equation}
The \emph{influence} of the set $T \subseteq [n]$ on $U$ is defined as
\begin{equation*}
    \inf_T[U] \coloneqq \sum_{\substack{x \in \Zfn \\ \operatorname{supp}(x) \cap T \neq \emptyset}} \abs{\widehat{U}(x)}^2.
\end{equation*}

\end{definition}
\begin{remark}[Monotonicity of influence]
 By \Cref{def_unitary_inf}, for any $S \subseteq T \subseteq [n]$ and $n$-qubit unitary $U$, we have $\inf_S[U] \le \inf_T[U]$.   
\end{remark}

\begin{definition}
    [Partial trace of unitaries]
    For any $n$-qubit unitary $U$ and $S \subseteq [n]$, the partial trace of $U$ on $\overline{S}$ is defined as
    \[
    \Tr_{\overline{S}}\p{U}  \coloneqq
    \sum_{i \in \{0,1\}^{\overline{S}}}\p{I_S \otimes \bra{i}} U \p{I_S \otimes \ket{i}}.
    \]
\end{definition}

\begin{definition}[$k$-junta unitary]\label{def:junta}
A unitary operator $U$ is called a \emph{$k$-junta unitary} if there exists $S \subseteq [n]$ with $|S| = k$ such that
\[
U = U_S \otimes I_{\overline{S}}
\]
for some $k$-qubit unitary $U_S$.
\end{definition}

A $k$-junta Boolean function is a function $f$ for which there are at most $k$ indices $i \in [n]$ such that $\inf_i(f) > 0$. Alternatively, $f$ is a $k$-junta if there exists a set $T \subseteq [n]$ of size $|T| \leq k$ such that $\inf_{\overline{T}}(f) = 0$. The same observations hold for unitary operators.
Here we empoly the following distance for unitary property testing introduced by Chen, Nadimpalli and Yuen~\cite{chen2023testing}.
\begin{definition}[Distance between unitaries]\label{def_dist_unitary}
Let $N\coloneqq 2^n$, for any $n$-qubit unitaries $U,V$,
\begin{equation}\label{eq:dist_UV}
\dist(U,V)  \coloneqq \min_{\theta} \frac{1}{\sqrt{2N}} || e^{i\theta} U-V||_F.
\end{equation}
Let $\mathcal{J}_T$ be the class of junta unitaries on $T$.
Therefore,
        \[
        \dist(U, \mathcal{J}_T ) 
         \coloneqq \min_{V \in \mathcal{J}_T} \dist(U, V).
        \]
        For $S \subseteq [n]$, we define
        \[       \dist(U,\mathcal{J}_{S,k}) \coloneqq\min_{T \in \binom{S}{k}} \dist(U, \mathcal{J}_T),
        \]
        where $\mathcal{J}_{S,k}$ are the class of $k$-junta unitaries on $S$.
        Similar to \Cref{def_dist_boolean}, we abbreviate $\mathcal{J}_{[n],k}$ as $\mathcal{J}_{k}$. 

\end{definition}
\begin{remark}
The distance between unitaries has an equivalent definition as follows. Consider
\begin{equation}\label{eq:UV_F}
    \begin{aligned}
        \frac{1}{N} \norm{U-V}_F^2
        & = \frac{1}{N} \Tr((U-V)^\dagger (U-V)) \\
        &= 2-\frac{1}{N} \Tr(U^\dagger V+V^\dagger U) \\
        &= 2\p{1-\frac{1}{N} \abs{\Tr  \p{U^{\dag}V}}},
        \end{aligned}
        \end{equation}
By \Cref{eq:dist_UV,eq:UV_F}, we have
\begin{equation}\label{eq:dist_U_V}
\dist(U,V) = \sqrt{1-\frac{1}{N}\abs{\Tr(U^{\dag}V)}}.
\end{equation}
\end{remark}

The following propositions demonstrate that any Boolean function can be viewed as a special case of a unitary.

\begin{proposition}[Proposition 9 in \cite{montanaro2008quantum}]\label{prop:fourier_boolean_unitary}
    For Boolean function $f:\{\pm1\}^n \rightarrow \{\pm1\}$, let $U_f$ be a diagonal matrix whose diagonal entries are the $2^n$ values of the function $f$. Then for any $x \in \mathbb{Z}_4^n$,
    \[
    \wU_f(x) = 
    \begin{cases}
        0, & \text{ if }x \notin \{0,3\}^n\\
        \widehat{f}(\operatorname{supp}(x)), & \text{ if } x \in \{0,3\}^n
    \end{cases}
    \]
\end{proposition}

As a corollary of \Cref{prop:fourier_boolean_unitary}, for any $i \in [n]$ and integer $1 \le k \le n$, 
\begin{equation}\label{eq:influence_boolean_unitary}
\inf_i^{\le k}[f] = \inf_{i}^{\le k}[U_f].
\end{equation}

\begin{proposition}[Equation (4.13) in \cite{chen2023testing}]
    For Boolean functions $f,g:\{\pm 1\}^n \rightarrow \{\pm1\}$, let $U_f$ and $U_g$ be diagonal matrices whose diagonal entries are the $2^n$ values of the function $f$ and $g$ respectively. If $g$ is a $k$-junta, then
    \[
\dist\p{U_f,U_g} = \sqrt{2\dist (f,g)}.
    \]
\end{proposition}

\begin{definition}[The \emph{Choi-Jamio{\l}kowski state} \cite{jamiolkowski1972linear, choi1975completely}]\label{def_choi_state} 
    Let $N \coloneqq 2^n$. The \emph{Choi-Jamio{\l}kowski state} of an $n$-qubit unitary $U$ is defined as
  \begin{equation*}
  \ket{v(U)}  \coloneqq (U \otimes I) \p{\frac{1}{\sqrt{N}} \sum_{0 \leq i < N} \ket{i}\ket{i}}  = \frac{1}{\sqrt{N}} \sum_{0 \leq i,j < N} U[i,j] \, \ket{i} \ket{j}.	
  \end{equation*}
  Since $U = \sum_{x\in \Zfn} \wU(x)\sigma_x$, we have $\ket{v(U)} = \sum_{x\in \Zfn} \wU(x)\ket{v(\sigma_x)}$. It is not hard to check that $\lrs{\ket{v(\sigma_x)}}_{x \in \Zfn}$ is an orthonormal basis.      
\end{definition}

%% file: sections/3_reduction.tex
In this section, we propose a coordinate extractor (\Cref{alg:ce}) to 
find a high-influence coordinate set $S \subseteq [n]$ such that $|S| = O(k^2)$, which is a key preprocessing step in the tolerant junta testing algorithms (\Cref{alg:tester,alg:tester2}). 
As a result, the number of coordinates we need to consider is reduced from $n$ to $O(k^2)$, making the subsequent tolerant testing step independent of the input size $n$.

\subsection{Algorithm and Analysis}

The Fourier Sampling Algorithm (Algorithm \ref{alg:fs}) is a core component of \Cref{alg:ce}. The state preparation step of Algorithm \ref{alg:fs} is inspired by  \cite{chen2023testing}, where $N \coloneqq 2^n$. Each time calling Algorithm \ref{alg:fs}, we can sample $x \in \Zfn$ with probability $|\wU(x)|^2$. Using Algorithm \ref{alg:fs} as a subroutine,  Algorithm \ref{alg:ce} can find high-influence coordinates with high probability as shown in \Cref{lemma:coor-ext}. 

\begin{algorithm}
\caption{\algname{Fourier-Sample}$(U)$}
\label{alg:fs}

\begin{algorithmic}[1]
\Input{Oracle access to $n$-qubit unitary $U$}	

\Output{$x\in \Zfn$ with probability $|\wU(x)|^2$}

\item Prepare the Choi-Jamiolkowski state
\[
\ket{v(U)} = (U \otimes I) \p{\frac{1}{\sqrt{N}} \sum_{0 \leq i < N} \ket{i}\ket{i}} 
\]
by querying $U$ once on the maximally entangled state;
	
\item Measure $\ket{v(U)}$ in the orthonormal basis $\{\ket{v(\sigma_x)}\}_{x \in \Zfn}$ as \Cref{def_choi_state}; 

\item Return the measurement result $x$.
\end{algorithmic}
\end{algorithm}

\begin{algorithm}
\caption{\algname{Coordinate-Extractor}$(U,k,\tau)$}
	\label{alg:ce}

\begin{algorithmic}[1]
\Input{Oracle access to $n$-qubit unitary $U$,  integer $1 \le k < n$, $0 < \tau < 1$}

\Output{$S \subseteq [n]$}

\item Initialize $e_i  \coloneqq 0$ for all $i \in [n]$, $\vep  \coloneqq 1/2$, $\delta  \coloneqq 0.01$;
\item Repeat the following for  $M \coloneqq\frac{2k}{\varepsilon^2\tau^2}\log\frac{k^2}{\delta \tau^2}$ times:
\begin{itemize}
	\item Sample $x \in \Zfn$ by running \algname{Fourier-Sample}$(U)$; 
	\item If $|x| \le k$, for all $i \in [n]$ s.t. $x_i \neq 0$, $e_i \leftarrow e_i+1$;   
\end{itemize}

\item Return $S  \coloneqq \lrs{i \in \lrb{n}: e_i \ge (1-\varepsilon)\frac{M\tau^2}{k}}$.
\end{algorithmic}
\end{algorithm}

\begin{lemma}\label{lemma:coor-ext}
	In \Cref{alg:ce}, $|S| \le {2k^2}/{\tau^2}$ and the probability of $S \supseteq \lrs{i \in \lrb{n}: \inf_i^{\le k}[U] \ge {\tau^2}/{k}}$ is at least $0.99$.
\end{lemma}
\begin{proof}

After running Step 2 of \Cref{alg:ce}, $\sum_{i \in [n]}e_i \le Mk$ because of the second bullet in Step 2. Thus, the number of $i$ such that 
\[e_i \ge (1-\varepsilon)\frac{M\tau^2}{k}
\]
is at most 
\[
\frac{k^2}{(1-\varepsilon)\tau^2} = \frac{2k^2}{\tau^2},
\]
which implies $|S| \le 2k^2/ \tau^2$.   

In each round of Step 2, 
$x \in \Zfn$ is sampled with probability $|\widehat{U}(x)|^2$. So combined with \Cref{eq:k-inf-u}, we know that for each $i$, the probability that $e_i$ increases by $1$ is 
\[
 \sum_{x \in \Zfn: |x| \le k,x_i \neq 0} |\wU(x)|^2 = \inf_i^{\le k}[U].
\]
Let $p_i  \coloneqq \inf_i^{\le k}[U]$. Then $e_i$  satisfies the binomial distribution $B(M,p_i)$. By the Chernoff bound, for those $i$ such that $p_i \ge {\tau^2}/{k}$,
\[
\begin{aligned}
\Prob{e_i \le (1-\varepsilon)\frac{M\tau^2}{k}} 
 &\le \Prob{\frac{e_i}{M} \le (1-\varepsilon)p_i} \\
 &\le \exp\left(-Mp_i\varepsilon^2/2\right)\\
 &\le 
\exp\left(-\log\frac{k^2}{\delta \tau^2}\right) \\
 &\le \frac{\delta\tau^2}{k^2}.
\end{aligned}
\]
Since 
\[
\begin{aligned}
	\sum_{i\in [n]}p_i &= \sum_{i \in [n]}\sum_{x \in \Zfn: |x| \le k,x_i \neq 0} |\wU(x)|^2 \\
	&\le k\sum_{x \in \Zfn: |x| \le k} |\wU(x)|^2 \le k,
\end{aligned}
\] 
the number of $i$ satisfying $p_i \ge {\tau^2}/{k}$ is at most ${k^2}/{\tau^2}$.
By the union bound, 
\[
\Prob{S \supseteq \lrs{i \in \lrb{n}: \inf_i^{\le k}[U] \ge \frac{\tau^2}{k}}} \ge 1-\frac{k^2}{\tau^2}\cdot \frac{\delta\tau^2}{k^2} = 1-\delta = 0.99.
\]
\end{proof}

\subsection{Properties of High-Influence Coordinates}
In the following, we present some properties of high-influence coordinates. These properties are useful to tolerant junta testing algorithms in \Cref{sec:gap,sec:gapless}, respectively.

\paragraph{Influence Approximation}
As shown in \Cref{lemma_inf_appro}, for a set $T'$ with size no more than $k$, if $T \subseteq T'$ contains high-influence coordinates of $T'$, then $\inf_{\overline{T'}}[U]$ can be approximated by $\inf_{\overline{T}}[U]$. \Cref{lemma_inf_appro} will be applicable in the tolerant junta tester with constant factor gap (\Cref{alg:tester}) in \Cref{sec:gap}.

In the following, we will 
use \Cref{lemma_inf_plus} to prove \Cref{lemma_inf_appro}.

\begin{lemma}
 \label{lemma_inf_plus}
 For any $n$-qubit unitary $U$, if $T' \subseteq [n]$ satisfies $|T'| \le k$, then for any $i \in T'$, 
 \[
 \inf_{\overline{T'} \cup \{i\}}[U] \le \inf_{\overline{T'}}[U]+\inf_i^{\le k}[U].
 \]
\end{lemma}  
\begin{proof}
\[
    \begin{aligned}
      \inf_{\overline{T'} \cup \{i\}}[U] &= \sum_{\substack{x \in \Zfn\\ \supp(x)\cap (\overline{T'} \cup \{i\}) \neq \emptyset }} |\widehat{U}(x)|^2 \\
      &= \sum_{\substack{x \in \Zfn\\ \supp(x)\cap \overline{T'} \neq \emptyset }} |\widehat{U}(x)|^2+\sum_{\substack{x \in \Zfn\\ \supp(x)\cap \overline{T'} = \emptyset \\ \supp(x) \cap \{i\} \neq \emptyset } } |\widehat{U}(x)|^2 \\
      &\le \inf_{\overline{T'}}[U]+\sum_{\substack{x \in \Zfn\\ |\supp(x)| \le k \\ x_i \neq 0 }} |\widehat{U}(x)|^2\\
      &= \inf_{\overline{T'}}[U]+\inf_i^{\le k}[U].
    \end{aligned}
    \]
\end{proof}

\begin{lemma}\label{lemma_inf_appro}
For any $n$-qubit unitary $U$, let $S\subseteq [n]$ satisfying
\[
S \supseteq \lrs{i \in [n]: \inf_i^{\le k}[U] \ge {\delta}/{k}},
\] 
then for any $T'\subseteq [n]$, $\abs{T'} \le k$, let $T = T'\cap S$, we have
\[
\inf_{\overline{T}}[U]  \le \inf_{\overline{T'}}[U]+\delta.
\]
\end{lemma}
\begin{proof}
By \Cref{lemma_inf_plus}, we have
    \begin{equation*}
        \begin{aligned}
        \inf_{\overline{T}}[U] 
        &\le \inf_{\overline{T'}}[U] + 
          \sum_{i\in T'\backslash T}\inf_i^{\le k}[U]\\
        &\le \inf_{\overline{T'}}[U] + k\cdot \f{\delta}{k}\\
        &= \inf_{\overline{T'}}[U] + \delta.
        \end{aligned}
    \end{equation*}
\end{proof}

\paragraph{Distance Approximation}
\Cref{lemma_high_influence_corr} shows the role of high-influence coordinates in approximating the correlation between Boolean functions and $k$-juntas. We generalize \Cref{lemma_high_influence_corr} to unitaries as \Cref{lemma_high_influence_dist}.
\begin{lemma}[Claim 29 in \cite{iyer2021junta}]\label{lemma_high_influence_corr}
For Boolean function $f:\{\pm1\}^n\rightarrow \{\pm1\}$ and subset $T \subseteq [n]$ with size $k$, let 
\[
S  \coloneqq \lrs{i \in T: \inf_i^{\le k}[f] \ge \f{\tau^2}{k}}, 
\]
where $\tau > 0$. Let $g:\{\pm1\}^n\rightarrow \{\pm1\}$ be a $k$-junta on $T \subseteq [n]$. Then  $\corr(f,\mathcal{J}_S) \ge \corr(f,g)-\tau$.   
\end{lemma}
\begin{remark}
    In \Cref{lemma_high_influence_corr}, since for any $k$-junta function $g$ on $T$, $\corr(f,\mathcal{J}_S) \ge \corr(f,g)-\tau$, it further implies 
    \[
    \corr(f,\mathcal{J}_S) \ge \corr(f,\mathcal{J}_T)-\tau.
    \]
    By \Cref{eq_relation_corr_dist}, we have \[
\dist(f,\mathcal{J}_S) \le \dist(f,\mathcal{J}_T)+\tau/2.
    \]
\end{remark}
\begin{lemma}[Unitary version of \Cref{lemma_high_influence_corr}] 
\label{lemma_high_influence_dist}
  Let $U$ be an $n$-qubit unitary, $T \subseteq [n]$ with $|T| = k$ and $\tau > 0$. If $S \subseteq [n]$ satisfies 
    \begin{equation*}
     S \supseteq\lrs{i\in T : \inf_{i}^{\le k}\lrb{U} \ge \f{\tau^2}{k}},
\end{equation*}
      then
$\dist\p{U,\mathcal{J}_S} \le \dist\p{U,\mathcal{J}_T}+\sqrt{\tau}$.
\end{lemma}



\begin{proof}
Let $V$ be a $k$-junta unitary on $T$ such that $\dist(U,V) = \dist(U,\mathcal{J}_T)$.
  We define 
  \[
  V'  \coloneqq\f{1}{2^{|\overline{S}|}} \Tr_{\overline{S}}\p{V}\otimes I_{\overline{S}}.
  \]
Let $N \coloneqq 2^n$. Since $V = \sum_{y \in \Zfn} \widehat{V}(y) \sigma_y$, for any $x \in \Zfn$, we have
  \[
  \begin{aligned}
  \widehat{V'}(x) &= \f{1}{N}\langle V', \sigma_x \rangle \\
  &= \f{1}{N }\sum_{y \in \Zfn}\widehat{V}(y)\langle \Tr_{\overline{S}}\p{\sigma_y} \otimes I_{\overline{S}} | \sigma_x \rangle.\\
  \end{aligned}
  \]
Note that
\[
\f{1}{N }\langle \Tr_{\overline{S}}\p{\sigma_y} \otimes I_{\overline{S}} | \sigma_x \rangle=
\begin{cases}
    1, &\text{if }\supp(y) \subseteq S\text{ and }y=x\\
    0, &\text{otherwise}
\end{cases}.
\]
  Thus, we have
  \begin{equation*}
    \widehat{V'}(x) = \begin{cases}
      \widehat{V}(x),  & \text{if } \supp(x)\subseteq S    \\
      0,               & \text{otherwise}
    \end{cases}.
  \end{equation*}
  Therefore,
  \begin{equation}  \label{eq:tr_UV_tau}
  \begin{aligned}
    \f{1}{N}\left|\left|\Tr\p{UV^\dagger}\right| - \left|\Tr\p{UV'^\dagger}\right|\right| &\le 
      \f{1}{N}\left|\Tr\p{UV^\dagger} - \Tr\p{UV'^\dagger}\right| \\
      &= \abs{\sum_{\substack{x\in \Zfn \\ \supp(x)\subseteq T \\ \supp(x) \nsubseteq S}} 
        \widehat{U}(x) \overline{\widehat{V}(x)}} \\
    &\le \sqrt{\sum_{\substack{x\in \Zfn \\ \supp(x)\subseteq T \\ \supp(x) \nsubseteq S}} 
      \abs{\widehat{U}(x)}^2} \\
    &\le \sqrt{\sum_{i\in T \backslash S}
      \sum_{\substack{x\in \Zfn \\ \supp(x)\subseteq T \\ \supp(x) \ni i}} 
      \abs{\widehat{U}(x)}^2} \\
    &\le \sqrt{\sum_{i\in T\backslash S} \inf_i^{\le k} [U]} \\
    &\le \sqrt{k\cdot \f{\tau^2}{k}} = \tau.
  \end{aligned}
  \end{equation}
Note that the spectral norm of $\Tr_{\overline{S}}(V)$ is 
\[
\begin{aligned}
    ||\Tr_{\overline{S}}(V)|| &= \norm{\sum_{i \in \{0,1\}^{\overline{S}}}\p{I_S\otimes \bra{i}} V \p{I_S\otimes \ket{i}}}\\
    &\le \sum_{i \in \{0,1\}^{\overline{S}}} \norm{\p{I_S\otimes \bra{i}} V \p{I_S\otimes \ket{i}}}\\
    &\le \sum_{i \in \{0,1\}^{\overline{S}}} \norm{I_S\otimes \bra{i}} \cdot \norm{V} \cdot \norm{I_S\otimes \ket{i}} \\
    &= 2^{|\overline{S}|}.
\end{aligned}
\]
Let the SVD decomposition of $\Tr_{\overline{S}}(V)/2^{|\overline{S}|}$ be $W_1 \Sigma W_2$. Then $\Sigma$ is a 
    diagonal matrix satisfying $0\le \Sigma \le I$ and $W_1,W_2$ are unitaries.
  There exists a probability distribution $\lrs{p_a}_{a\in \lrs{\pm 1}^S}$ such that 
    $\Sigma = \sum_{a\in \lrs{\pm 1}^S} p_a \diag(a)$, where $\diag(a)$ is the diagonal matrix 
    with $a$ in the diagonal entries.
  Therefore, 
  \[
  \f{1}{2^{|\overline{S}|}}\Tr_{\overline{S}}[V] = W_1 \Sigma W_2 = \sum_{a\in \lrs{\pm 1}^S} p_aW_1\diag(a)W_2.
  \]
  For any $a \in \{\pm 1\}^S$, let $V_a = W_1\diag(a) W_2 \otimes I_{\overline{S}}$. 
  We conclude that $V' = \sum_{a\in \lrs{\pm 1}^S} p_a V_a$. 
  Thus,
  \begin{equation}
    \label{eq:tr_dist}
    \abs{\Tr\p{UV'^\dagger}} = \abs{\sum_{a\in \lrs{\pm 1}^S}p_a\Tr\p{UV_a^\dagger}}
      \le \sum_{a\in \lrs{\pm 1}^S}p_a\abs{\Tr\p{UV_a^\dagger}}.
  \end{equation}
  By \Cref{eq:tr_UV_tau,eq:tr_dist},
  there exists some $b\in \lrs{\pm 1}^S$ such that 
  \begin{equation}\label{eq:tr_UV}
    \f{1}{N}\abs{\Tr\p{UV_b^\dagger}} 
      \ge \f{1}{N}\abs{\Tr\p{UV'^\dagger}} 
      \ge \f{1}{N}\abs{\Tr\p{UV^\dagger}} - \tau.
  \end{equation}
    Noted that $V_b$ is a junta on $S$. Combining \Cref{eq:dist_U_V,eq:tr_UV}, we have
  \[
  \dist^2\p{U,V_b} \le \dist^2\p{U,V}+\tau.
  \]
    Therefore,
    \[
    \begin{aligned}
      \dist\p{U,\mathcal{J}_S} &\le  \dist\p{U,V_b}\\ 
      &\le
      \dist\p{U,V}+\sqrt{\tau}\\
      &=\dist\p{U,\mathcal{J}_T}+\sqrt{\tau}. \end{aligned}
    \]
\end{proof}
As a corollary of \Cref{lemma_high_influence_dist}, \Cref{cor_high_influence_dist} shows that the high-influence coordinator extractor can be used to estimate the distance between a unitary and $k$-juntas:

\begin{corollary}
\label{cor_high_influence_dist}
  Let $U$ be an $n$-qubit unitary and $\tau > 0$. For set $S \subseteq [n]$, if
  \begin{equation*}
    S \supseteq 
    \lrs{i\in [n] : \inf_{i}^{\le k}\lrb{U} \ge \f{\tau^2}{k}},
  \end{equation*}
  then
$\dist\p{U,\mathcal{J}_{S,k}} \le \dist\p{U,\mathcal{J}_k}+\sqrt{\tau}$.
\end{corollary}
\begin{proof}
    By \Cref{lemma_high_influence_dist}, for any $T \subseteq [n]$ with $|T| = k$, 
    \[
    \dist(U, \mathcal{J}_{S\cap T}) \le \dist(U, \mathcal{J}_T)+\sqrt{\tau}.
    \]
    Thus, 
    \[
    \begin{aligned} \dist(U,\mathcal{J}_{S,k}) &=
    \min_{T \in \binom{[n]}{k}}\dist(U,\mathcal{J}_{S\cap T})\\
    &\le \min_{T\in \binom{[n]}{k}}\dist(U,\mathcal{J}_T)+\sqrt{\tau} \\
    &= \dist(U,\mathcal{J}_k)+\sqrt{\tau}. 
    \end{aligned}
    \]
\end{proof}

Combining \Cref{lemma:coor-ext} 
and \Cref{cor_high_influence_dist}, we give the following theorem, which will be used to design a tolerant junta tester for an arbitrary gap (\Cref{alg:tester2}) in \Cref{sec:gapless}.
\begin{theorem}\label{theorem_reduction}
   For any $n$-qubit unitary $U$, integer number $1 \le k < n$ and $\tau > 0$, using $O\p{\frac{k}{\tau^2}\log\frac{k}{ \tau}}$ queries,  \Cref{alg:ce} can find a subset $S \subseteq [n]$ such that $|S| \le {2k^2}/{\tau^2}$. Besides, with probability at least $0.99$,  
    \[    \dist\p{U,\mathcal{J}_{S,k}} \le \dist\p{U,\mathcal{J}_{k}}+\sqrt{\tau}.
    \]
\end{theorem}

%% file: sections/4_gapped_tester.tex
In this section, we will present an efficient tolerant junta tester for both Boolean functions and unitaries: when $\vep_1$ and $\vep_2$ are separated by a constant multiplicative factor, tolerant $k$-junta testing can be achieved using a quantum algorithm with $O(k \log k)$ queries. For Boolean functions, this is a substantial improvement over the current classical exponential results \cite{nadimpalli2024optimal, chen2025mysterious}, indicating that quantum computing has an exponential speedup for the tolerant junta testing problem. 

Moreover, the construction of our algorithm is natural and concise. We first apply the high-influence coordinate extractor (\Cref{alg:ce}) in \Cref{sec:ce} to narrow down the range of possible influential set $T$ so that the number of rounds of the inner loop is independent of $n$. Then we use $\algname{Fourier-Sample}$ (\Cref{alg:fs}) to estimate the influence for each subset with size $k$. 

We will first show the relation between the influence and distance from juntas. Then we present our tolerant junta tester for unitary operators. For the Boolean case, the algorithm and analysis follow a similar line with slightly different parameters.

\subsection{Relationships between Distances-to-juntas and Influences}

For both Boolean functions $f$ and unitary operators $U$, their distances to juntas are closely related to the influence of underlying sets. Most of these facts were established in previous work (see \Cref{prop:ff,prop:uc,prop:uf}).
We also give a self-contained proof for \Cref{prop:fc}.

\begin{lemma}\label{prop:fc}
    If a Boolean function $f:\{\pm 1\}^n \rightarrow \{\pm1\}$ is $\varepsilon_1$-close to a $k$-junta on $T \subseteq [n]$ with $|T|=k$, then
\[
\inf_{\overline{T}}[f] \leq 4 \varepsilon_1.
\]
\end{lemma}

\begin{proof}
    Assume that $f$ is $\varepsilon_1$-close to a $k$-junta Boolean function $g$ on $T$. For all $S \nsubseteq T$, we have $\widehat{g}(S) =0$. Therefore,
    \[
    \sum_{S\subseteq T} \widehat{g}(S)^2 =1.
    \]
    We know that
    \[
    \dist(f, g)=\frac{1-\langle f, g\rangle}{2} \leq \vep_1,
    \]
    so
    \[
    \sum_{S \subseteq T} \widehat{f}(S) \widehat{g}(S) =\langle f, g \rangle=1-2\dist(f, g) \geq 1 - 2\vep_1.
    \]
    By Cauchy-Schwarz inequality, 
    \[
    \sum_{S \subseteq T} \widehat{f}(S)^2 \cdot 1\geq \left| \sum_{S \subseteq T} \widehat{f}(S)\widehat{g}(S) \right|^2 \geq (1 - 2\vep_1)^2 > 1 - 4\vep_1.
    \]
    Therefore,
    \[
    \inf_{\overline{T}}[f] =\sum_{S \cap \overline{T} \neq \emptyset} \widehat{f}(S)^2  \leq 4 \varepsilon_1.
    \]
\end{proof}

\begin{lemma}[Proposition 2.3 of \cite{blais2009testing}] \label{prop:ff}
    If a Boolean function $f:\{\pm 1\}^n \rightarrow \{\pm1\}$ is $\varepsilon_2$-far from every $k$-junta, then for every set $T \subseteq [n]$ of size $|T| \leq k$,  
\[
\inf_{\overline{T}}[f]\geq \varepsilon_2.
\]
\end{lemma}

\begin{lemma}[Lemma 2.8 of \cite{chen2024tolerant}]\label{prop:uc}
    If a unitary $U$ is $\varepsilon_1$-close to a $k$-junta unitary on $T \subseteq [n]$ with $|T| = k$, then
\[
\inf_{\overline{T}}[U] \leq 2 \varepsilon_1^2.
\]
\end{lemma}

\begin{lemma}[\cite{wang2011property}, proof of Theorem 4]\label{prop:uf}
   If a unitary $U$ is $\varepsilon_2$-far from every $k$-junta unitary, then for every $T \subseteq [n]$ with $|T| \leq k$, 
\[
\inf_{\overline{T}}[U] \geq \frac{\varepsilon_2^2}{4}.
\]
\end{lemma}

\subsection{An 	\texorpdfstring{$O(k\log k)$}{klogk}-query Non-adaptive Tolerant Junta Tester}\label{sec:u}

Now we are ready to show our non-adaptive tolerant junta tester for unitary operators.

\begin{algorithm}
\caption{$\algname{Tolerant-Junta-Tester}(U, \vep_1, \vep_2)$}
\label{alg:tester}

\begin{algorithmic}[1]\label{alg:1}

\Input{Oracle access to $n$-qubit unitary $U$, 
  integer $1 \le k < n$, $0<2\sqrt{2}\vep_1 < \vep_2<1$}

\Output{`Yes' or `No'}

\item Let $\delta  \coloneqq (\vep_2^2/4 - 2\vep_1^2)/{3}$;
 \item Run $S = \algname{Coordinate-Extractor}\p{U, k, \tau}$
  where $\tau  \coloneqq \sqrt{\delta}$; \label{item_ce}
\item For all \(i \in [M]\), run $\algname{Fourier-Sample}(U)$
  to obtain $x^{(i)}$, with 
  $M \coloneqq \log {\p{200 \p{\binom{\abs{S}}{k}+1}}}/2{\delta^2}$;
\item If $\abs{S} \ge k$, repeat the following for each subset $T$ 
  of $S$ with $\abs{T} = k$; otherwise let $T=S$:
\begin{itemize}
    \item Initialize $s_{\overline{T}}  \coloneqq 0$;
    \item \(i = 1, 2, \dots,M\), 
      if 
      $\supp(x^{(i)}) \cap \overline{T} \neq \emptyset$,
      $s_{\overline{T}} \leftarrow s_{\overline{T}}+1$;  
    \item Let $s_{\overline{T}} 
      \leftarrow {s_{\overline{T}}}/{M}$;
\end{itemize}
\item If $\min_T s_{\overline{T}} \leq 2\vep_1^2+2\delta$, 
  output `Yes', otherwise output `No'.
\end{algorithmic}
\end{algorithm}

\begin{theorem}[Formal version of \Cref{thm:i2}]\label{thm:tester}
    For any unitary $U$, integer $1 \le k < n$ and constant numbers  $\vep_1, 
      \vep_2 \in (0,1)$ satisfying $ \vep_2 >  2\sqrt{2}\vep_1$, if $U$ is $\vep_1$-close
      to some $k$-junta unitary, 
      $\algname{Tolerant-Junta-Tester}(U, \vep_1, \vep_2)$ 
      will output `Yes' with probability $9/10$. 
    On the other hand, if $U$ is $\vep_2$-far from any $k$-junta 
      unitary, $\algname{Tolerant-Junta-Tester}(U, \vep_1, \vep_2)$ 
      will output `No' with probability $9/10$.
    Besides, the algorithm is non-adaptive and makes $O\p{\f{k}{\delta^2} \log \f{k}{\delta}}$ queries to $U$, where $\delta = (\vep_2^2-8\vep_1^2)/12$. 
\end{theorem}

\begin{proof}
    From \Cref{lemma:coor-ext}, $\abs{S}\le {2k^2}/{\delta}$ and with probability 
      at least $0.99$, we have     \begin{equation*}
       S \supseteq 
        \lrs{i \in \lrb{n}: \inf_i^{\le k}[U] \ge \frac{\delta}{k}}.
    \end{equation*}

    If $U$ is $\vep_1$-close to some $k$-junta unitary,
      according to \Cref{prop:uc}, there exists 
      $T'\subseteq [n]$ with $|T'| = k$ such that 
      $\inf_{\overline{T'}}[U] \le 2 \varepsilon_1^2$.
    Let $T= T'\cap S$, by \Cref{lemma_inf_appro}, we have
    \begin{equation*}
    \inf_{\overline{T}}[U] 
        \le \inf_{\overline{T'}}[U] + \delta
        \le 2 \varepsilon_1^2 + \delta.
    \end{equation*}
    We conclude that if $U$ is $\vep_1$-close to some 
      $k$-junta unitary, there must exist a $T\subseteq S$
      with $\abs{T} \le k$, such that 
      $\inf_{\overline{T}}[U] \le 2 \varepsilon_1^2 + \delta$.

    If $U$ is $\vep_2$-far from any $k$-junta unitary,
      according to \Cref{prop:uf}, we have that
      for any $T\subseteq [n]$ with $\abs{T}=k$,      $\inf_{\overline{T}}[U] \ge \varepsilon_2^2/4$.
    Therefore for any $T\subseteq {S}$ with $\abs{T}=k$,
      $\inf_{\overline{T}}[U] \ge \varepsilon_2^2/4$.
      
    For any subset $T\subseteq S$ with $\abs{T} = k$, 
      and also $T=S$, we have 
        $\E{s_{\overline{T}}} = \inf_{\overline{T}}[U]$.
    By Hoeffding's bound, 
    \begin{equation*}
        \Pr  \lrb{\abs{ 
          s_{\overline{T}} - \inf_{\overline{T}}[U] }
          \geq \delta} 
        \leq 2 \exp\left( -2M\delta^2 \right) 
        \leq \f{0.01}{\binom{\abs{S}}{k}+1}.
    \end{equation*}
    By the union bound, we conclude that, with probability at least
      $0.99$, 
    \begin{equation*}
      \abs{  s_{\overline{T}} - \inf_{\overline{T}}[U] }
          \le \delta
    \end{equation*}
    holds for all $T\subseteq S$ with $\abs{T}\le k$
      and $T=S$.

    Therefore, if there exists a $T\subseteq S$ with
      $\abs{T}\le k$ such that 
      $\inf_{\overline{T}}[U] \le 2 \varepsilon_1^2+\delta$,
      with probability at least $0.99$, the algorithm
      will find a $T\subseteq S$ with $\abs{T}=k$ or $T=S$ such that 
      $s_{\overline{T}} \le 2 \varepsilon_1^2+2\delta$
      and output `Yes' according to the mononicity of 
      influence.
    On the other case, if for any $T\subseteq {S}$ with $\abs{T}=k$,
      $\inf_{\overline{T}}[U] \ge \varepsilon_2^2/4$, 
      with probability at least $0.99$ the algorithm will output
      `No'.

Note that the number of queries in \algname{Coordinate-Extractor} (\Cref{alg:ce}) is $O\p{\f{k}{\tau^2} \log \f{k}{\tau^2}}$.
Since $\tau = \sqrt{\delta},|S| \le 2k^2/\delta$,
    the query complexity of \Cref{alg:tester} is
    \[
    \begin{aligned}
    O\p{\f{k}{\tau^2} \log \f{k}{\tau^2}}+O\p{\f{\log\binom{|S|}{k}}{\delta^2}} &= O\p{\f{k}{\delta}\log \f{k}{\delta}}+O\p{\f{k\log \f{|S|}{k}}{\delta^2}} \\
    &= O\p{\f{k}{\delta^2}\log \f{k}{\delta}}.
    \end{aligned}
    \]
    To see that it is non-adaptive, we note that this algorithm 
      could invoke $\algname{Fourier-Sample}$ at the beginning
      and then feed the date into the subroutine
      $\algname{Coordinate-Extractor}$ to identify $S$,
      and then estimate the influence for all size-$k$ 
      subsets of $S$.
\end{proof}

\paragraph{For Boolean functions}

The junta tester for Boolean functions follows a similar 
  line. 
Most part of \Cref{alg:tester} works well for Boolean
  functions, except that we need to make some changes to the parameters. 
According to \Cref{prop:fc,prop:ff}, 
  we replace $2\vep_1^2$ and $\vep_2^2/4$ by $4\vep_1$ and $\vep_2$ in the 
  first and the fifth line of the algorithm, and therefore 
  $\delta \coloneqq (\vep_2 - 4\vep_1)/{3}$.
We call the algorithm with new parameters $\algname{Tolerant-Boolean-Junta-Tester}$.
We have the following theorem.

\begin{theorem}[Formal version of \Cref{thm:i1}]\label{thm:tester-boolean}
    For any Boolean function $f:\{\pm1\}^n \rightarrow \{\pm1\}$, integer $1 \le k < n$ and constant numbers $\vep_1, 
      \vep_2 \in (0,1)$ satisfying $\vep_2 > 4\vep_1$, if $f$ is $\vep_1$-close
      to some $k$-junta Boolean function, 
      $\algname{Tolerant-Boolean-Junta-Tester}(f, \vep_1, \vep_2)$ 
      will output `Yes' with probability $9/10$; if $f$ is $\vep_2$-far from any $k$-junta Boolean
      function, $\algname{Tolerant-Boolean-Junta-Tester}(f, \vep_1, \vep_2)$ 
      will output `No' with probability $9/10$.
    Besides, the algorithm is non-adaptive and makes $O\p{\f{k}{\delta^2} \log \f{k}{\delta}}$ queries to $f$, where $\delta = (\vep_2-4\vep_1)/3$.
\end{theorem}

%% file: sections/5_gapped_tester_local.tex
In this section we introduce \algname{Influence-Sample} as \Cref{alg:is} 
which was proposed by Bao and Yao \cite{bao2023testing},
and then 
  show that we can replace the invocation of \algname{Fourier-Sample}
  with \algname{Influence-Sample} to obtain the following theorem. 

\begin{algorithm}
\caption{\algname{Influence-Sample}$(U)$}
\label{alg:is}

\begin{algorithmic}[1]\label{alg:2}
\Input{Oracle access to $n$-qubit unitary $U$.}	
\Output{$x \in \lrs{0,1}^n $.}
    \item Uniformly randomly choose $y\in \{0,1\}^n$. 
        Prepare state $\ket{y}$;
    
    \item Uniformly randomly choose $V$ from 
      \begin{center}
      $\lrs{
          I = \p{\begin{matrix}1&0\\0&1\end{matrix}},
          H = \f{1}{\sqrt{2}}\p{\begin{matrix}1&1\\1&-1\end{matrix}},
          R = \f{1}{\sqrt{2}}\p{\begin{matrix}1&-i\\-i&1\end{matrix}}
      }$;
      \end{center}
    
    \item Obtain 
      $\p{V^{\otimes n}}^\dagger U V^{\otimes n} \ket{y}$;
    
    \item Measure the qubits in the computational basis.
      Let the result be $y'$;

    \item Return $x\coloneqq  y \oplus y'$, where $\oplus$ is 
      the XOR operator.
\end{algorithmic}
\end{algorithm}

\begin{theorem}[Formal version of \Cref{thm:i4}]\label{thm:tester-l}
    For any $n$-qubit unitary $U$, integer $1 \le k < n$, and constant numbers 
      $\vep_1, \vep_2 \in (0,1)$ satisfying $ 2\sqrt 3\vep_1 < \vep_2$, 
      if $U$ is $\vep_1$-close to some $k$-junta unitary, 
      $\algname{Tolerant-Junta-Tester-Local}(U, k, \vep_1, \vep_2)$ 
      will output `Yes' with probability $9/10$. 
    On the other hand, if $U$ is $\vep_2$-far from any $k$-junta 
      unitary, $\algname{Tolerant-Junta-Tester-Local}(U, k, 
      \vep_1, \vep_2)$ will output `No' with probability $9/10$.
    Besides, the algorithm is non-adaptive, only 
      requires single-qubit operations, does not require ancilla qubits, and makes 
$O\p{\f{k}{\delta^2} \log \f{k}{\delta}}$
queries to $U$, where $\delta = (\vep_2^2-12\vep_1^2)/18$.
\end{theorem}

Here we assume the gap $\vep_2>2\sqrt{3}\vep_1$, which is  larger than $\vep_2>2\sqrt{2}\vep_1$ as required for \Cref{alg:tester}. However, \Cref{alg:is} only involves single-qubit operations and does not require ancilla qubits, while \Cref{alg:tester} involves \algname{Fourier-sample} which requires preparing multiple copies of EPR states. Indeed, \algname{Influence-Sample} has recently been implemented on a photonic platform~\cite{zhan2025}.

Before proving \Cref{thm:tester-l}, we need the following key property of \algname{Influence-Sample}.
The proof is essentially from \cite{bao2023testing}, for self-containment we also include the proof
    in \Cref{appendixB}.

\begin{restatable}{lemma}{lemmais}
\label{lemma:is}
    For any $n$-qubit unitary $U$ and $T\subseteq [n]$, 
      let $x=\algname{Influence-Sample}\p{U}$, we have
    \begin{equation}\label{eq:lemmaisI}
      \Prob{x_T\neq 0} \le \inf_T\lrb{U} \le \f{3}{2}\Prob{x_T\neq 0}.
    \end{equation}
    Besides, if we remove the use of $R$ gate in 
      the algorithm, we have
    \begin{equation}\label{eq:lemmaisII}
      \Prob{x_T\neq 0} \le \inf_T\lrb{U} 
        \le 2\Prob{x_T\neq 0}.
    \end{equation}
\end{restatable}


\begin{algorithm}
\caption{\algname{Coordinate-Extractor-Local}$(U,k,\tau)$}
  \label{alg:isl}

\begin{algorithmic}[1]
\Input{Oracle access to $n$-qubit unitary $U$, integer $1 \le k < n$,
   $0 < \tau < 1$}
\Output{$S \subseteq [n]$.}
\item Initialize $e_i  \coloneqq 0$ for all $i\in [n]$, 
  $\vep \coloneqq  1/2$, $\delta \coloneqq  0.01$;
\item Repeat the following for $M \coloneqq\f{2k}{\vep^2\tau^2}
  \log\f{k^2}{\delta \tau^2}$ times:
\begin{itemize}
	\item Sample $x \in \lrs{0,1}^n$ by running 
      \algname{Influence-Sample}$(U)$;
	\item If $|x| \le k$, for all $i \in [n]$ s.t. $x_i \neq 0$, 
      $e_i \leftarrow e_i+1$;   
\end{itemize}

\item For $i\in [n]$, let $e_i \leftarrow {e_i}/{M}$;

\item Return $S \coloneqq  \lrs{i \in \lrb{n}: e_i \ge (1-\vep)\f{\tau^2}{k}}$.
\end{algorithmic}
\end{algorithm}

\begin{theorem}\label{thm:coor-ext-l}
	In \Cref{alg:isl}, $|S| \le {2k^2}/{\tau^2}$, 
      and the probability of $S \supseteq \lrs{i \in \lrb{n}: 
      p_i \ge {\tau^2}/{k}}$ is at least $1-\delta$,
      where
      \[
      p_i \coloneqq \Prob{x = \algname{Influence-Sample}\p{U},
        \abs{x}\le k, x_i\neq 0}.
      \]  
\end{theorem}

\begin{proof}
    Note that $\E{e_i} = p_i$ holds for all $i \in [n]$.
    Following the same line as 
    \Cref{lemma:coor-ext},
it is sufficient to show $\sum_{i=1}^n p_i\le k$ as follows,
    \begin{equation*}
        \sum_{i=1}^n p_i = \E{\sum_{i=1}^n e_i}
          \le \f{1}{M} \cdot Mk =k.
    \end{equation*}
\end{proof}

\begin{algorithm}
\caption{\algname{Tolerant-Junta-Tester-Local}$(U, k, \vep_1, \vep_2)$}
\label{alg:tester-l}

\begin{algorithmic}[1]

\Input{Oracle access to unitary $U$, 
  integer $1 \le k < n$, $0 < 2\sqrt{3}\vep_1 < \vep_2 < 1$}
\Output{`Yes' or `No'}
\item Let $\delta  \coloneqq {(\vep_2^2/6 - 2\vep_1^2)}/{3}$;
\item Run $S\leftarrow \algname{Coordinate-Extractor-Local}\p{
  U, k, \tau}$ with $\tau \coloneqq\sqrt{\delta}$;
\item For $\forall i \in [M]$, run $\algname{Influence-Sample}\p{U}$
  to obtain $x^{(i)}$, with 
  $M \coloneqq \log {\p{200 \p{\binom{\abs{S}}{k}+1}}}/2{\delta^2}$;
\item If $\abs{S}\ge k$, repeat the following for each subset 
  $T$ of $S$ with $\abs{T} = k$; otherwise let $T=S$:
\begin{itemize}
    \item Initialize $s_{\overline{T}}  \coloneqq 0$;
    \item For $i = 1, 2, \dots,M$, 
      if
      $\supp(x^{(i)}) \cap \overline{T} \neq \emptyset$,  
    $s_{\overline{T}} \leftarrow s_{\overline{T}}+1$;  
    \item Let $s_{\overline{T}} 
      \leftarrow {s_{\overline{T}}}/{M}$;
\end{itemize}
\item If $\min_{T\subseteq S} s_{\overline{T}} \leq 2\vep_1^2+2\delta$, return 
  `Yes', otherwise return `No'.
\end{algorithmic}
\end{algorithm}

Our tolerant junta tester with single-qubit operations is \Cref{alg:tester-l}.
We are now ready to prove its property.

\begin{proof}[Proof of \Cref{thm:tester-l}]
    Let $x = \algname{Influence-Sample}\p{U}$.
    The proof follows a similar line as \Cref{thm:tester},
      the main difference is that here we will use 
      a probability regarding $x$ to replace the $\inf^{\le k}$
      notation in the proof of \Cref{thm:tester}.
      
    From \Cref{thm:coor-ext-l} we know that with probability 
      at least $0.99$, we have $\abs{S}\le {2k^2}/{\delta}$ and

    \begin{equation*}
        S \supseteq \lrs{i \in \lrb{n}: 
          p_i \ge \frac{\tau^2}{k}},
    \end{equation*}
    where $p_i = \Prob{\abs{x}\le k, x_i\neq 0}$.
    
    If $U$ is $\vep_1$-close to some $k$-junta unitary,
      according to \Cref{prop:uc}, there exists 
      $T'\subseteq [n]$ such that 
      $\inf_{\overline{T'}}[U] \le 2 \varepsilon_1^2$.
    Let $T= T'\cap S$, combining with \Cref{lemma:is}, we have 
    \begin{equation*}
        \Prob{x_{\overline{T}} \neq 0} 
          \le \Prob{x_{\overline{T'}} \neq 0} 
            + \sum_{i\in T'\backslash S} \Prob{\abs{x}\le k, x_i\neq 0}
          \le \inf_{\overline{T'}}[U] + k \cdot \f{\tau^2}{k}
          \le 2 \varepsilon_1^2 + \delta.
    \end{equation*}
    We conclude that there must exist a $T\subseteq S$ with 
      $\abs{T}\le k$, such that 
      $\Prob{x_{\overline{T}} \neq 0}  \le 2 \varepsilon_1^2 + \delta$.

    If $U$ is $\vep_2$-far from any $k$-junta unitary,
      according to \Cref{prop:uf}, we have that
      for any $T\subseteq [n]$ with $\abs{T}\le k$,      
      $\inf_{\overline{T}}[U] \ge \varepsilon_2^2/4$.
    Therefore for any $T\subseteq S$ with $\abs{T}\le k$,
      combining with \Cref{lemma:is},
    \begin{equation*}
      \Prob{x_{\overline{T}} \neq 0} \ge 
        \f{2}{3} \inf_{\overline{T}}[U] \ge \f{\vep_2^2}{6}.
    \end{equation*}
    The rest of the proof is similar to the proof 
      of \Cref{thm:tester}.


    It is easy to see the algorithm is  
      non-adaptive, requires only single-qubit operations, and makes $O\p{k\log k}$ queries.
\end{proof}

\paragraph{For Boolean functions}

Most part of \Cref{alg:tester-l} works well for Boolean
  functions, except that we need to make some changes to the parameters. 
According to \Cref{prop:fc,prop:ff}, 
  we replace $2\vep_1^2$ and $\vep_2^2/6$ by $4\vep_1$ and $2\vep_2/3$
  in the first and the fifth line of the algorithm, and therefore 
  $\delta \coloneqq (2\vep_2/3 - 4\vep_1)/{3}$.
We call the algorithm with new parameters \algname{Tolerant-Boolean-Junta-Tester}.
We have the following theorem.

\begin{theorem}[Formal version of \Cref{thm:i3}]
    \label{thm:tester-boolean-l}
    For any Boolean function $f:\{\pm1\}^n\rightarrow \{\pm 1\}$, integer $k < n$, and constant numbers 
      $\vep_1, \vep_2$ satisfying $6\vep_1 < \vep_2$, if $f$ is 
      $\vep_1$-close to some $k$-junta function, 
      $\algname{Tolerant-Boolean-Junta-Tester-Local}(f, k, \vep_1, \vep_2)$ 
      will output `Yes' with probability $9/10$. 
    On the other hand, if $f$ is $\vep_2$-far from any $k$-junta 
      function, $\algname{Tolerant-Boolean-Junta-Tester}(f, k, \vep_1, \vep_2)$ 
      will output `No' with probability $9/10$.
    Besides, the algorithm is non-adaptive, only 
      requires single-qubit operations, does not require ancilla qubits, and makes 
      $O\p{\f{k}{\delta^2} \log \f{k}{\delta}}$
      queries to $f$, where $\delta = (2\vep_2-12\vep_1)/9$.
\end{theorem}

%% file: sections/6_ungapped_tester.tex
In this section, as an application of \Cref{alg:ce}, we propose  \Cref{alg:tester2} to solve the tolerant junta testing problem for any $\epsilon_1,\epsilon_2$. Specifically, our result is summarized in the following theorem, and its proof is postponed to the end of this section. 
\begin{theorem}[Formal version of \Cref{thm:i5}]\label{thm:general}
For $n$-qubit unitary $U$, integer $1 \le k < n$ and $0 < \vep_1 < \vep_2 < 1$,
if $U$ is $\vep_1$-close
      to some $k$-junta unitary, 
      $\algname{Gapless-Tolerant-Junta-Tester}(U, \vep_1, \vep_2)$ 
      will output `Yes' with probability 0.98. On the other hand, if $U$ is $\vep_2$-far
      from any $k$-junta unitary, 
      $\algname{Gapless-Tolerant-Junta-Tester}(U, \vep_1, \vep_2)$ 
      will output `No' with probability 0.98. Besides, the number of queries in the algorithm is $2^{O(k \log k/\vep)}$, where $\vep  \coloneqq \varepsilon_2-\varepsilon_1$.
\end{theorem}


Specifically, the idea of  \Cref{alg:tester2} is first to find a high-influence coordinate set $S$ by \Cref{alg:ce}, and then estimate the distance $\dist(U,\mathcal{J}_T)$ between the unitary and $k$-juntas on $T$ for any $T \subseteq S$ with $|T| = k$ by \Cref{alg:we}.

\subsection{Subroutine: Estimating 	\texorpdfstring{$\dist(U,\mathcal{J}_T)$}{dist}}

As \Cref{def_dist_unitary}, for any $n$-qubit unitary $U$ and $T \subseteq [n]$ with $|T| = k$, 
\begin{equation}\label{eq:dist_U_JT}
    \begin{aligned}
        \dist(U, \mathcal{J}_T ) 
        & = \min_{V \in \mathcal{J}_T} \dist(U, V) \\
        & = \min_V \min_{\theta} \frac{1}{\sqrt{2N}} || e^{i\theta} U-V||_F,
    \end{aligned}
    \end{equation}
    where $N=2^n$. Assuming that $k$-junta unitary $V$ satisfies $V=V_T \otimes I_{\overline{T}}$, we have
    \begin{equation}\label{eq:UV_F_2}
    \begin{aligned}
        \frac{1}{N} || U-V||_F^2
        & = 2-\frac{1}{N} \Tr(U^\dagger V+V^\dagger U) \\
        & = 2-\frac{1}{N} \Tr(\Tr_{\overline{T}}(U^\dagger V)+\Tr_{\overline{T}}(V^\dagger U)) \\
        & = 2-\frac{1}{N} \Tr((\Tr_{\overline{T}}U)^\dagger V_T+(\Tr_{\overline{T}}U) V_T^\dagger),
    \end{aligned}
    \end{equation}
    where the first equality comes from \Cref{eq:UV_F}.
    We abbreviate  $\Tr_{\overline{T}}U$ as $U_T$. Combining \Cref{eq:dist_U_JT,eq:UV_F_2}, we have
    \begin{align*}
        \dist(U, \mathcal{J}_T ) 
        & = \sqrt{ \frac{1}{2}\min_V \left(2-\frac{1}{N} \Tr(U_T^\dagger V_T+U_T V_T^\dagger)\right) } \\
        & = \sqrt{ 1- \frac{1}{2N}\max_V \left( \Tr(U_T^\dagger V_T+U_T V_T^\dagger)\right)}.
    \end{align*}
    Suppose the SVD-decomposition of $U_T$ is ADB, we choose $V_T = AB$ to obtain the maximum
    \begin{align*}
        \Tr(U_T^\dagger V_T+U_T V_T^\dagger) =\Tr(D^\dagger+D) =2 \sum_i \sigma_i(D) = 2||U_T||_*,
    \end{align*}
    here $||U_T||_*$ denotes the nuclear norm (or Schatten 1-norm) of $U_T$. Therefore,
    \begin{equation}\label{eq:dist}
        \dist(U, \mathcal{J}_T ) =\sqrt{ 1- \frac{1}{N} ||U_T||_*}.
    \end{equation}
Then we estimate $\dist(U, \mathcal{J}_T )$ using \Cref{alg:we}. 
The sketch idea of \Cref{alg:we} is first to construct an approximate matrix $\widetilde{U}_T$ by estimating every entry of $U_T$ using the Hadamard test in \Cref{lemma_hadamard_test}, and then compute the estimation value  $\widetilde{\dist}(U, \mathcal{J}_T) \coloneqq \sqrt{1-||\widetilde{U}_T||_*/N}$.

\begin{algorithm}
    \caption{\algname{Warmup-Estimator($U,T,\tau,\delta$)}}
    	\label{alg:we}
    \begin{algorithmic}[1]
    \Input{Oracle access to $n$-qubit unitary $U$, $T \subseteq [n]$ with $|T| = k, 0 < \tau, \delta < 1$}

\Output{Estimation value of $\dist(U,\mathcal{J}_T)$}

    \item Let $M \coloneqq\frac{2^{k+1}}{\tau^2} \ln\left(\frac{4}{\delta}\right)$;
    \item For any $m \in [M]$: 
    \begin{itemize}
        \item Sample $l_m \in \{0,1\}^{\overline{T}}$  uniformly;
        \item For each $i, j \in \{0,1\}^{T}$, using $O\p{\frac{2^k\log (1/\delta)}{\tau^2}}$ controlled applications of $U$, to obtain estimation value $\hat{X}_{m,i,j}$ s.t. 
        \[
        \Pr\left[|\hat{X}_{m,i,j}-\bra{il_m} U\ket{jl_m}|\le \frac{\tau}{2^{{k}/{2}+1}}\right] \ge 1-\delta/2
        \]
        by the Hadamard test in \Cref{lemma_hadamard_test};
    \end{itemize}
    \item Construct a matrix $\widetilde{U}_T \in M_{2^k \times 2^k}$ by setting 
    \[
    \widetilde{U}_T[i][j]  \coloneqq\frac{2^{n-k}}{M} \sum_{m=1}^M \hat{X}_{m,i,j};
    \]
    \item Compute and return $\widetilde{\dist}(U, \mathcal{J}_T) \coloneqq\sqrt{1-\|\widetilde{U}_T\|_*/N}$.
    \end{algorithmic}
    \end{algorithm}

       \begin{lemma}[Hadamard test, Theorem 2.3 in \cite{luongo2022quantum}]\label{lemma_hadamard_test}
        Suppose $U$ is an $n$-qubit unitary and $x,y \in \{0,1\}^n $,  
        there is a quantum algorithm to estimate $\langle x|U|y\rangle$ with additive error $\tau$ using $O\p{{\log(1/\delta)}/{\tau^2}}$ times controlled applications of $U$, with probability at least $1-\delta$.
    \end{lemma}

In the following, we first give \Cref{lemma_norm_relation} and then analyze the query complexity and success probability of \Cref{alg:we} by \Cref{thm_estimator_US}.

\begin{lemma}\label{lemma_norm_relation}
    If $A, B \in \mathbb{C}^{m\times m}$ satisfy $|A[i][j] - B[i][j]| \leq \tau$ for any $i,j$, then $|\|A\|_* - \|B\|_*|\leq m^{3/2} \tau.$    
    \end{lemma}
    
    \begin{proof}
    Let $E \coloneqq A - B$. Then for any $i,j$, $|E[i][j]| =| A[i][j] - B[i][j]|\leq \tau$. By the triangle inequality,
    \[
    |\|A\|_* - \|B\|_*| = |\|B + E\|_* - \|B\|_*| \leq \|E\|_*.
    \]
    By Cauchy's inequality,
    \begin{align*}
    \|E\|_*&=\sum_i \sigma_i \\
    &\leq \sqrt{\left( \sum_i 1\right)\left( \sum_i \sigma_i^2 \right)} \\
    & = \sqrt m \|E\|_F \\
    &= \sqrt{m\sum_{i,j} |E[i][j]|^2} \\
    &\leq \sqrt{m\sum_{i,j} \tau^2} = \sqrt{m^3 \cdot \tau^2} = m^{3/2} \tau.
    \end{align*}
    \end{proof}

\begin{lemma}\label{thm_estimator_US}
The number of queries in \Cref{alg:we} is ${2^{O(k)}\log^2(1/\delta)}/{\tau^4}$.
    Besides, the output of \Cref{alg:we} satisfies
    \begin{equation*}
        |\widetilde{\dist}(U,\mathcal{J}_T)-\dist(U,\mathcal{J}_T)|\leq \sqrt{\tau}
    \end{equation*}
    with probability at least $1-\delta$.
\end{lemma}

\begin{proof}

In \Cref{alg:we}, the number of queries is 
\[
\frac{2^{k+1}}{\tau^2} \ln\left(\frac{4}{\delta}\right) \cdot 2^k \cdot 2^k \cdot O\left(\frac{2^k\log (1/\delta)}{\tau^2}\right) = \f{2^{O(k)}\log^2(1/\delta)}{\tau^4}.
\]
    Next, we analyze the success probability of the algorithm. Note that for any $i,j \in \{0,1\}^{T}$, 
    \begin{equation*}
    U_T[i][j]= \bra{i}U_T\ket{j} =\sum_{l\in \{0,1\}^{\overline{T}}}\bra{il} U\ket{jl}. 
    \end{equation*}
    Let $X_{m,i,j}  \coloneqq \bra{il_m} U\ket{jl_m}$. 
    Then $U_T[i][j] = 2^{n-k}\mathbb{E}_m[X_{m,i,j}]$. 
    Let $\tau'={\tau}/\p{2^{k/2+1}}$. By Hoeffding's bound,
    \begin{equation*}
    \Pr\left( \left| \frac{1}{M} \sum_{l=1}^M X_{m,i,j} - \mathbb{E}_m[X_{m,i,j}] \right| \geq \tau' \right) \leq 2 \exp\left( -2M\tau'^2 \right) \leq \delta/2.
    \end{equation*}
Moreover, we have
\[
\Pr\left( \left| \frac{1}{M} \sum_{l=1}^M \hat{X}_{m,i,j} - \frac{1}{M} \sum_{l=1}^M X_{m,i,j}] \right| \geq \tau' \right) \le \delta/2 
\]
   Thus, 
   \[
   \Pr\left( \left| \frac{1}{M} \sum_{l=1}^M \hat{X}_{m,i,j} - \mathbb{E}_m[X_{m,i,j}] \right| \leq 2\tau' \right) \ge 1-\delta, 
   \]
which implies that    with probability at most $1 - \delta$, we have
    \begin{equation*}
        \f{1}{N}\left|\widetilde{U}_T[i][j]-U_T[i][j]\right| \leq \frac{2\tau'}{2^k}=\frac{\tau}{2^{3k/2}}, \ \forall i,j \in \{0,1\}^k.
    \end{equation*}
By \Cref{lemma_norm_relation}, we have \[
\f{1}{N}|\|\widetilde{U}_T\|_* - \|U_T\|_*|\leq \tau,
\]
    which implies
        \[
        | \widetilde{\dist}^2(U, \mathcal{J}_T )-\dist^2(U, \mathcal{J}_T )|=\f{1}{N}|\|\widetilde{U}_T\|_* - \|U_T\|_*|\leq \tau. 
        \]
        Thus,
        \begin{align*}
        | \widetilde{\dist}(U, \mathcal{J}_T )-\dist(U, \mathcal{J}_T )| &\leq \frac{\tau}{\widetilde{\dist}(U, \mathcal{J}_T )+\dist(U, \mathcal{J}_T )} \\
        &\leq \frac{\tau}{ | \widetilde{\dist}(U, \mathcal{J}_T )-\dist(U, \mathcal{J}_T )|}.
    \end{align*} 
    As a result, we have
    \begin{equation*}
        |\widetilde{\dist}(U, \mathcal{J}_T )-\dist(U, \mathcal{J}_T )| \leq \sqrt\tau.
    \end{equation*}
    \end{proof}

\subsection{Main Algorithm}

Based on \Cref{alg:ce,alg:we}, \Cref{alg:tester2} is shown as follows.

\begin{algorithm}
\caption{\algname{Gapless-Tolerant-Junta-Tester}$(U, \vep_1, \vep_2)$}
\label{alg:tester2}

\begin{algorithmic}[1]
\Input{Oracle access to $n$-qubit unitary $U$, 
  integer $1 \le k < n$, $0 < \vep_1 < \vep_2 < 1$}

\Output{`Yes' or `No'}
\item Let $\vep\coloneqq \vep_2-\vep_1$, $\tau \coloneqq \vep^2/16$ and $\delta \coloneqq0.01$; 
        \item Run $S \leftarrow \algname{Coordinate-Extractor}(U,k,\tau)$;
        \item For every subset of $T \subseteq S$ with size $k$, run $\widetilde{\dist}(U, \mathcal{J}_T) \leftarrow \algname{Warmup-Estimator($U,T,\tau,\delta'$)}$, where $\delta' \coloneqq \delta/\binom{|S|}{k}$;
            
    \item Let
    \begin{equation*}
        \widetilde{\dist}(U, \mathcal{J}_{S,k})  \coloneqq\min_{T\subseteq \binom{S}{k}} \widetilde{\dist}(U, \mathcal{J}_T).  
    \end{equation*}
    If $\widetilde{\dist}(U, \mathcal{J}_{S,k}) \le \vep_1+\vep/2$, output `Yes', otherwise output `No'.
\end{algorithmic}
\end{algorithm}

In the following, we prove \Cref{thm:general} by analyzing the query complexity and the success probability of \Cref{alg:tester2}.

\begin{proof}[Proof of \Cref{thm:general}]
Note that the number of queries in \algname{Coordinate-Extractor} (\Cref{alg:ce}) is $O\p{\f{k}{\tau^2} \log \f{k}{\tau}}$ and $|S| = O({k^2}/{\tau^2})$. 
By \Cref{thm_estimator_US}, the number of queries in \algname{Warmup-Estimator} (\Cref{alg:we}) is
\[
\f{2^{O(k)}\log^2(1/\delta)}{\tau^4}.
\]
Thus, the query complexity of  \Cref{alg:tester2} is
\[
\begin{aligned}
O\p{\f{k}{\tau^2}\log \f{k}{\tau}} + \f{2^{O(k)}}{\tau^4} \cdot \log^2 \f{\binom{|S|}{k}}{\delta}\cdot \binom{|S|}{k} 
&=O\p{\f{k}{\tau^2}\log \f{k}{\tau}} + \f{2^{O(k)}}{\tau^4} \cdot O\p{k^2\log^2 \f{|S|}{k}} \cdot 2^{O\p{k \log \f{|S|}{k}}}\\ 
&=\f{O\left(k\log \f{k}{\tau}\right)}{\tau^2}+\f{2^{O(k)}\cdot O\p{k \log \f{k}{\tau^2}} \cdot 2^{O\left(k\log \f{k}{\tau^2}\right)}}{\tau^4}\\
&= 2^{O\left(k \log \f{k}{\tau}\right)}\\
&=2^{O\left(k \log \f{k}{\vep}\right)},
\end{aligned}
\]
where $\vep= \vep_1-\vep_2$.

For the success probability of \Cref{alg:tester2}, it suffices to analyze the probability of 
\[
\abs{\widetilde{\dist}(U,\mathcal{J}_{S,k}) - \dist(U,\mathcal{J}_{k})}\le \vep/2.
\]
By \Cref{thm_estimator_US}, after running Step 3 of \Cref{alg:tester2}, with probability $1-\delta'$, for all $T \subseteq S$ such that $|T| = k$, we have 
\[
\abs{\widetilde{\dist}(U,\mathcal{J}_T) - \dist(U,\mathcal{J}_T)} \le \sqrt{\tau}.
\]
Since
\[
\begin{aligned}
    \widetilde{\dist}(U, \mathcal{J}_{S,k})
&=\min_{T\subseteq \binom{S}{k}} \widetilde{\dist}(U, \mathcal{J}_T),\\
{\dist}(U, \mathcal{J}_{S,k})
&=\min_{T\subseteq \binom{S}{k}} {\dist}(U, \mathcal{J}_T),
\end{aligned}
\]
by the union bound, the probability of
\[
\abs{\widetilde{\dist}(U, \mathcal{J}_{S,k})-{\dist}(U, \mathcal{J}_{S,k})} \le \sqrt{\tau}
\]
is at least $1-\delta = 0.99$.
Moreover, by \Cref{theorem_reduction}, with probability at least 0.99, we have
\[
\abs{\dist(U,\mathcal{J}_{S,k}) - \dist(U,\mathcal{J}_{k})}\le \sqrt{\tau}.
\]
Therefore, we have
\[
|\widetilde{\dist}(U,\mathcal{J}_{S,k}) - \dist(U,\mathcal{J}_{k})|\le 2\sqrt{\tau} = \vep/2
\]
with probability at least $0.99^2 > 0.98$.
\end{proof}

%% file: sections/7_conclusion.tex
In this paper, we present quantum algorithms for tolerant testing of $k$-junta Boolean functions and unitary operators, achieving significant breakthroughs in both fields. For Boolean functions, our work fills a crucial gap by providing the first non-trivial quantum tester. Although there are constraints on parameters, it demonstrates the potential for an exponential quantum speedup over classical algorithms for this problem. For unitary operators, our algorithms achieve an exponential improvement in query complexity to $O(k \log k)$, surpassing previous quantum methods for this problem when a constant gap is present. Finally, we provide a more general tester for unitaries that works for any parameter gap.

A key contribution is the Coordinate-Extractor subroutine, which efficiently identifies a small set of influential coordinates. This allows the query complexity of our algorithms to be independent of the input size $n$. We also introduce an experiment-friendly variant that requires only single-qubit operations, eliminating the need for complex multi-qubit gates and ancilla qubits. 

In summary, this work marks a significant advancement in quantum property testing, particularly by highlighting the immense potential of Fourier analysis of unitaries. The techniques developed here not only provide new tools for quantum learning theory but also pave the way for the development of more powerful and practical quantum algorithms for complex computational tasks.

%% file: sections/app1_lowerbound.tex
In this section, we will prove the following lower bound,
  using similar techniques as \cite{chen2024mildly}.

\lowerboundI*


Let $n=k+a$ and $A$ be a subset of $[n]$ of size $a$.
Let $C\coloneqq [n]\setminus A$. We refer to the variables $x_i$ for $i\in C$ as \emph{control variables} and the variables $x_i$ for $i\in A$ as \emph{action variables}. We first define two pairs of functions over $\{0,1\}^A$ on the action variables as follows (we will use these functions later in the definition of $\dy$ and $\dn$):
Let $h^{(+,0)},h^{(+,1)},h^{(-,0)}$ and $h^{(-,1)}$ be Boolean functions over $\{0,1\}^A$ defined as follows.
Here $c_1$ and $c_2$ are parameters we will decide later.

\begin{equation*}
h^{(+,0)}(x_A) =
	\begin{cases}
       \ 0 & |x_A|>\frac{a}{2}+c_1; \\[0.5ex]
       \ 0 & |x_A|\in[\frac{a}{2}-c_1,\frac{a}{2}+c_1]; \\[0.5ex]
       \ 0 & |x_A|<\frac{a}{2}-c_1.
    \end{cases}
\quad h^{(+,1)}(x_A) =
	\begin{cases}
       \ 1 & |x_A|>\frac{a}{2}+c_1; \\[0.5ex]
       \ 0 & |x_A|\in[\frac{a}{2}-c_1, \frac{a}{2}+c_1]; \\[0.5ex]
       \ 1 & |x_A|<\frac{a}{2}-c_1.
    \end{cases}
\end{equation*}
and
\begin{equation*}
h^{(-,0)}(x_A) =
	\begin{cases}
       \ 1 & |x_A|>\frac{a}{2}+c_1; \\[0.5ex]
       \ 0 & |x_A|\in[\frac{a}{2}-c_1,\frac{a}{2}+c_1]; \\[0.5ex]
       \ 0 & |x_A|<\frac{a}{2}-c_1.
    \end{cases}
\quad h^{(-,1)}(x_A) =
	\begin{cases}
       \ 0 & |x_A|>\frac{a}{2}+c_1; \\[0.5ex]
       \ 0 & |x_A|\in[\frac{a}{2}-c_1,\frac{a}{2}+c_1]; \\[0.5ex]
       \ 1 & |x_A|<\frac{a}{2}-c_1.
    \end{cases}
\end{equation*}

To draw a function $\fy\sim\dy$, we first sample a set $A \subseteq [n]$ of size $k$ uniformly 
  at random, set $C=[n]\setminus A$, and sample a Boolean function $r$ over $\{0,1\}^{C}$ uniformly at random. 
Then the Boolean function $\fy$ over $\{0,1\}^n$ is defined as follows:
\begin{equation*}\label{eq:fyes}
\fy(x)=
\begin{cases}
h^{(+,0)}(x_A)  & r(x_{C})=0\\[0.3ex]
h^{(+,1)}(x_A)  & r(x_{C})=1
\end{cases}
\end{equation*}
To draw $\fn\sim \dn$, we first sample  $A$ and $r$ in the same way as in $\dy$,
  and $\fn$ is defined as
\begin{equation*}\label{eq:fno}
\fn(x)=\begin{cases}
h^{(-,0)}(x_A)  & r(x_{C})=0\\[0.3ex]
h^{(-,1)}(x_A)  & r(x_{C})=1
\end{cases}
\end{equation*}

\begin{lemma}[Lemma 16 of \cite{chen2024mildly}, modified]
  \label{lemma:classical_close}
  If $c_1\le 0.1\sqrt{a}$, we have
    every function in the support of $\dy$ is 
    $\f{2c_1}{\sqrt{a}}$-close to a $k$-junta.
\end{lemma}

\begin{proof}
    Given $A$ and $r$, let $f$ be the corresponding function 
        constructed from $\dy$, then it's easy to see 
        $f$ is close to $C$-junta.
    The distance between $f$ and $C$-junta is at most
    \begin{equation*}
        \f{1}{2^a}\sum_{x=\f{a}{2}-c_1}^{\f{a}{2}+c_1}\binom{x}{a} 
        \le \f{1}{2^a} 2c_1 \binom{\f{a}{2}}{a} \le \f{2c_1}{\sqrt{a}}
    \end{equation*}
    where for the last inequality we use Fact 4 of \cite{chen2024mildly}
        that $\f{1}{2^a}\binom{\f{a}{2}}{a} \le \f{1}{\sqrt{a}}$.
\end{proof}

\begin{lemma}[Lemma 17 of \cite{chen2024mildly}, modified]
  \label{lemma:classical_far}
  If $c_1\le 0.1\sqrt{a}$, we have
    with probability at least $1-o_n(1)$, 
    $\fn \sim \dn$ is $0.2$-far from any $k$-junta.
\end{lemma}

\begin{proof}
    First of all we need Chernoff bound and a union bound that
        with probability at least $1-o_n(1)$, $A$ and $f$ satisfy that:
        for every $i\in C$, there are at least $0.24\cdot 2^k$
        many strings $x\in \lrs{0,1}^C$ with $x_i=0$, such that 
        $r(x)\neq r\p{x^{(i)}}$, where $x^{(1)}$ is the string in
        $\lrs{0,1}^C$ obtained by flip $i$-th bit of $x$.
        
    We assume below that $A$ and $f$ satisfy the above condition, and 
        we will show that such $f$ constructed from $\dn$ is $0.2$-far
        from any $k$-junta boolean function.
        
    Let $g$ be an $I$-junta boolean function where $I\subseteq [n]$ and
        $\abs{I}=k$. 
    Given the condition above, if $I\neq C$, the number of bichromatic 
        edges of $f$ along each direction $i\in C$ is at least
    \begin{equation*}
        0.24\cdot 2^k \cdot (1-0.1)2^a \ge 0.2\cdot 2^n
    \end{equation*}
    and therefore the distance between $f$ and $g$ is at least $0.2$.
    On the other hand if $I=C$, then the distance between $f$ and $g$, 
        by the definition, is at least $0.45$.
\end{proof}

\begin{lemma}[Lemma 18 of \cite{chen2024mildly}, modified]
  For any non-adaptive classical algorithm, 
    querying $t$ times the function sampled from 
    $\dy$ and $\dn$ , the probability it can
    distinguish them is at most
    $t^2\cdot \ff{\binom{n-2c_1}{k}}{\binom{n}{k}}
      \le t^2\cdot \p{1-\f{2c_1}{n}}^k
      \le t^2\cdot e^{\ff{-2c_1k}{n}}$.
\end{lemma}

\begin{proof}
    To distinguish $\dy$ from $\dn$, the algorithm must query
        to $x$ and $y$ such that $x_C=y_C$,
        $\abs{x_A} > \f{a}{2}+c_1$ and
        $\abs{y_A} < \f{a}{2}-c_1$.
    For any $x$ and $y$, we must have $\abs{x} - \abs{y} > 2c_1$
        to make this happen, and the probability that
        $C$ satisfies the above condition is at most
        $\ff{\binom{n-2c_1}{k}}{\binom{n}{k}}$.
    To conclude, the probability that the algorithm queries to such $x$
        and $y$ is at most $t^2\cdot \ff{\binom{n-2c_1}{k}}{\binom{n}{k}}$.
\end{proof}

\begin{proof}[Proof of \Cref{thm:lowerbound1}]
  To prove \Cref{thm:lowerbound1}, we choose 
    $a=k$, $n=2k$ and $c_1=0.005\sqrt{k}$.
  It follows that $2^{\Omega\p{\sqrt{k}}}$ 
    queries are necessary.

\end{proof}

%% file: sections/app2_proof_of_influence_sample.tex
In this section, \Cref{{lemma:is}} is restated and proved as follows.

\lemmais*

\begin{proof}
    We will prove \Cref{eq:lemmaisI} here, it's similar to show \Cref{eq:lemmaisII}.
    Recall that in \Cref{alg:is} we will uniformly randomly choose 
        $y \in \lrs{0,1}^n$ and $V$ from 
    \begin{center}
      $\lrs{
          I = \p{\begin{matrix}1&0\\0&1\end{matrix}},
          H = \f{1}{\sqrt{2}}\p{\begin{matrix}1&1\\1&-1\end{matrix}},
          R = \f{1}{\sqrt{2}}\p{\begin{matrix}1&-i\\-i&1\end{matrix}}
      }$,
    \end{center}
    therefore we have
    \begin{align*}
         &\,\Prob{x_T = 0 \mid V=I} \\
        =&\, \f{1}{2^n}\sum_{i\in \{0,1\}^T}\sum_{j\in \{0,1\}^{\overline{T}}} 
         \Prob{S\cap T_1 =\varnothing\mid y=i_Tj_{\overline{T}},V=I} \\
        =&\, \f{1}{2^n}\sum_{i}\sum_{j} 
         \Tr\,\lrb{\p{\bra{i}_T\otimes I_{\overline{T}}}\cdot U\ketbra{i,j})U^\dagger\cdot \p{\ket{i}_T\otimes I_{\overline{T}}}} \\
        =&\, \f{1}{2^n}\sum_i \sum_j \sum_{y,z\in \Zfn} 
         \Tr\,\lrb{\bra{i}\otimes I \cdot \wU(y)\overline{\wU(z)}\sigma_y \ketbra{i,j} \sigma_z \cdot \ket{i}\otimes I} \\
        =&\, \f{1}{2^n}\sum_i \sum_j \sum_{y,z\in \Zfn} 
          \wU(y)\overline{\wU(z)} \bra{i} \sigma_{y_T} \ket{i} \bra{i}\sigma_{z_T} \ket{i} \cdot \bra{j} \sigma_{z_{\overline{T}}}\sigma_{y_{\overline{T}}} \ket{j} \\
        =&\, \f{1}{2^n}\sum_{y,z\in \Zfn} 
         \wU(y)\overline{\wU(z)}  \p{\sum_i  \bra{i} \sigma_{y_T} \ket{i} \bra{i}\sigma_{z_T} \ket{i}} \cdot \p{\sum_j \bra{j} \sigma_{z_{\overline{T}}}\sigma_{y_{\overline{T}}} \ket{j}}.
    \end{align*}
    For the sum in the first parentheses, notice that when 
        $y_T\notin \{0,3\}^T$, we have
        $\bra{i} \sigma_{y_T} \ket{i}=0$, $\forall i$.
    Similar for $z_T$.
    Furthermore, note that when $y_T\neq z_T\in \{0,3\}^S$,
        it always holds that $\sum_i  \bra{i} \sigma_{y_T} \ket{i} \bra{i}\sigma_{z_T} \ket{i}=0$,
        otherwise it's $2^{\abs{T}}$.
    For the sum in the second parentheses, since
     $\sum_j \bra{j} \sigma_{z_{\overline{T}}}\sigma_{y_{\overline{T}}} \ket{j}
         = \Tr\lrb{\sigma_{z_{\overline{T}}}\sigma_{y_{\overline{T}}}}$,
    we know that if $y_{\overline{T}}\neq z_{\overline{T}}$, the sum is $0$, otherwise it's $2^{n-\abs{T}}$.

    To conclude, we have
    \begin{align*}
    &\,\Prob{x_T = 0 \mid V=I}\\
        =&\, \f{1}{2^n}\sum_{y,z\in \Zfn} 
         \wU(y)\overline{\wU(z)}  \p{\sum_i  \bra{i} \sigma_{y_T} \ket{i} \bra{i}\sigma_{z_T} \ket{i}} \cdot \p{\sum_j \bra{j} \sigma_{z_{\overline{T}}}\sigma_{y_{\overline{T}}} \ket{j}} \\
        =&\, \sum_{y\in \Zfn; y_T\in \{0,3\}^{T}} \abs{\wU(y)}^2.
    \end{align*}

    Note that $H^\dagger \sigma_0 H = \sigma_0$, $H^\dagger \sigma_1 H = \sigma_3$,
     $H^\dagger \sigma_2 H = -\sigma_2$, $H^\dagger \sigma_3 H = \sigma_1$,
     and $R_x^\dagger \sigma_0 R_x = \sigma_0$, $R_x^\dagger \sigma_1 R_x = \sigma_1$,
     $R_x^\dagger \sigma_2 R_x = -\sigma_3$, $R_x^\dagger \sigma_3 R_x = \sigma_2$.
    From a similar calculation, we have
    \begin{align*}
        &\Prob{x_T = 0  \mid V=H} = 
            \sum_{y\in \Zfn; y_{T}\in \{0,1\}^{T}} \abs{\wU(y)}^2, \\
        &\Prob{x_T = 0  \mid V=R_x} = 
            \sum_{y\in \Zfn; y_{T}\in \{0,2\}^{T}} \abs{\wU(y)}^2.
    \end{align*}

    Recall the definition of the influence,
    $$
        1 - \inf_{T}[U] = \sum_{y\in \Zfn; y_{T}=0}\abs{\wU(y)}^2.
    $$
    We have 
    \begin{equation*}
        \sum_{y\in \Zfn; y_{T}=0}\abs{\wU(y)}^2 \le \f{1}{3}\sum_{V'\in \lrs{I,H,R_x}}
            \Prob{x_T = 0  \mid V=V'} = \Prob{x_T = 0},
    \end{equation*}
    and therefore $\inf_{T}[U] \ge \Prob{x_T \neq 0}$.
    On the other hand, by carefully comparing the terms we are summing here, we observe that
    \begin{align*}
        2\inf_{T}[U] = 2\sum_{y\in \Zfn; y_{T}\neq 0}\abs{\wU(y)}^2 
            &\le \p{\sum_{y\in \Zfn; y_{T}\notin \{0,1\}^{T}} + \sum_{y\in \Zfn; y_{T}\notin \{0,2\}^{T}} 
                + \sum_{y\in \Zfn; y_{T}\notin \{0,3\}^{T}}} \abs{\wU(y)}^2 \\
            &= \sum_{V'\in \lrs{I,H,R_x}} \Prob{x_T \neq 0  \mid V=V'} = 3 \Prob{x_T \neq 0 },
    \end{align*}
    and therefore $\f{2}{3}\inf_T[\Phi] \le  \Prob{x_T\neq 0}$.
    We complete the proof for \Cref{eq:lemmaisI}.
\end{proof}